\documentclass[conference]{IEEEtran}
\IEEEoverridecommandlockouts
\usepackage[normalem]{ulem}
\newcommand{\stkout}[1]{\ifmmode\text{\sout{\ensuremath{#1}}}\else\sout{#1}\fi}
\usepackage{amsmath,amssymb,amsthm,amsfonts}
\usepackage{algorithmic}
\usepackage{graphicx}
\usepackage{textcomp}
\usepackage[table,RGB]{xcolor}
\usepackage[numbers,sort&compress]{natbib}
\usepackage[ruled,vlined,linesnumbered]{algorithm2e}
\usepackage{wrapfig}
\usepackage{multirow}
\usepackage{array}
\usepackage[scriptsize]{subfigure}
\usepackage{flushend}

\newtheorem{Theorem}{Theorem}

\theoremstyle{definition}
\newtheorem{mydef}{Definition}
\let\oldnl\nl
\newcommand{\nonl}{\renewcommand{\nl}{\let\nl\oldnl}}
\newcommand{\Break}{\textbf{break}}
\SetKwProg{Fn}{Function}{}{}
\SetKw{Continue}{continue}

\def\BibTeX{{\rm B\kern-.05em{\sc i\kern-.025em b}\kern-.08em
    T\kern-.1667em\lower.7ex\hbox{E}\kern-.125emX}}
\begin{document}

\title{Adaptive Encoding Strategies for \\Erasing-Based Lossless Floating-Point Compression}

\author{
	\IEEEauthorblockN{Ruiyuan Li~$^{1}$, Zheng Li~$^{1}$, Yi Wu~$^{1}$, Chao Chen~$^{1}$, Tong Liu~$^{2}$, Yu Zheng~$^{3, 4}$}\thanks{Chao Chen is the corresponding author.}
	\IEEEauthorblockA{
		$^{1}$~\textit{College of Computer Science, Chongqing University, China}\\
		$^{2}$~\textit{School of Computer Engineering and Science, Shanghai University, China}\\
		$^{3}$~\textit{JD iCity, JD Technology, China} \quad $^{4}$~\textit{JD Intelligent Cities Research, China}\\
		\{ruiyuan.li, zhengli, wu\_yi, cschaochen\}@cqu.edu.cn, tong\_liu@shu.edu.cn, msyuzheng@outlook.com
	}
}
\maketitle

\begin{abstract}
Lossless floating-point time series compression is crucial for a wide range of critical scenarios. Nevertheless, it is a big challenge to compress time series losslessly due to the complex underlying layouts of floating-point values. The state-of-the-art erasing-based compression algorithm {\em Elf} demonstrates a rather impressive performance. We give an in-depth exploration of the encoding strategies of {\em Elf}, and find that there is still much room for improvement. In this paper, we propose {\em Elf}$^*$, which employs a set of optimizations for leading zeros, center bits and sharing condition. 
Specifically, we develop a dynamic programming algorithm with a set of pruning strategies to compute the adaptive approximation rules efficiently. We theoretically prove that the adaptive approximation rules are globally optimal. We further extend {\em Elf}$^*$ to Streaming {\em Elf}$^*$, i.e., {\em SElf}$^*$, which achieves almost the same compression ratio as {\em Elf}$^*$, while enjoying even higher efficiency in streaming scenarios. We compare {\em Elf}$^*$ and {\em SElf}$^*$ with 8 competitors using 22 datasets. The results demonstrate that {\em SElf}$^*$ achieves 9.2\% relative compression ratio improvement over the best streaming competitor while maintaining similar efficiency, and that {\em Elf}$^*$ ranks among the most competitive batch compressors.  All source codes are publicly released.
\end{abstract}

\section{Introduction}\label{sec:introduction}


The proliferation of sensing devices~\cite{zhang2023wifi} and IoTs (Internet of Things)~\cite{li2015internet, nguyen20216g} has triggered off an unprecedented sheer volume of floating-point time series data (abbr. time series or time series data). These time series data are usually generated and transferred in a streaming fashion. For example, a single flight of a Boeing 787 aircraft can generate up to 0.5 terabytes of time series data~\cite{jensen2017time}. If these data are transmitted and stored in their raw formats, it would not only waste a lot of bandwidth and storage cost, but also slow down the transmission and storage efficiency. To cope with this issue, it is universally acknowledged to compress these data before transmitting or storing them~\cite{chiarot2023time}. Compared with lossy compression~\cite{lazaridis2003capturing, liang2022sz3, lindstrom2014fixed, liu2021decomposed, liu2021high, zhao2021optimizing, zhao2022mdz, liu2021exploring, liu2022dynamic, barbarioli2023hierarchical, li2023lossy, gong2022region, jiao2022toward} that may lose some important information, lossless time series compression~\cite{pelkonen2015gorilla, liakos2022chimp, li2023elf, Elf+, blalock2018sprintz, burtscher2007high, burtscher2008fpc, jensen2018modelardb, jensen2021scalable, ratanaworabhan2006fast, yu2020two, vestergaard2020randomly, gomez2021lossless} plays a fundamental role in a wide range of critical scenarios, such as scientific calculation~\cite{laporta2000high}, network transmission and data management~\cite{li2021trajmesa, li2020trajmesa, li2020just, ShardingSphere, he2022trass, xiao2022time, Yu2021distributed, bao2016managing, li2017cloud, john2023application, wang2022frequency, zhang2023compressstreamdb, mao2023morphstream, adams2020monarch}. For example, any error in a flight could lead to a disastrous consequence. Besides, it is unacceptable for users if their data are tampered with during transmission or storage.

However, it is a big challenge to losslessly compress streaming time series data, because the underlying structure of floating-point values in time series is rather complex. As shown in Fig.~\ref{fig:ElfXOR}(a), in accordance with IEEE 754 Standard~\cite{kahan1996lecture}, a double-precision floating-point value occupies 64 bits, in which 1 bit is for the \textbf{Sign}, 11 bits for the \textbf{Exponent}, and 52 bits for the \textbf{Mantissa}. To achieve streaming lossless compression of time series, one representative method is based on the XORing operation. As shown in Fig.~\ref{fig:ElfXOR}(a), given a double-precision floating-point value $v_t = 3.17$ and its predecessor $v_{t-1} = 3.25$, we perform an XORing operation on them, i.e., $xor_t = v_t \oplus v_{t-1}$. Since two consecutive values in a time series do not vary much, the XORed result $xor_t$ is expected to contain many \textbf{leading zeros}, and may also contain many \textbf{trailing zeros}. We can record the numbers of leading zeros and trailing zeros with very few bits, and store the \textbf{center bits} as they are, thus achieving lossless compression. When decompressing $v_t$, we first restore $xor_t$, and then perform another XORing operation, i.e., $v_t = xor_t \oplus v_{t-1}$ (note that $v_{t-1}$ is already decompressed in streaming scenarios).

\begin{figure}[t] 
  \centering
  \includegraphics[width=3.6in]{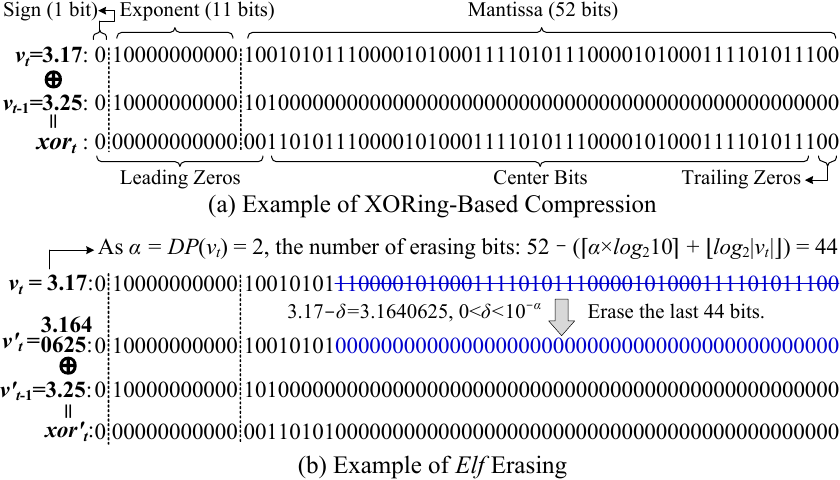}%
\vspace{-5pt}
  \caption{Main Idea of {\em Elf} Algorithm.}
  \label{fig:ElfXOR}
  \vspace{-15pt}
\end{figure}

Chimp~\cite{liakos2022chimp} finds that in most cases the XORed result actually contains very few trailing zeros, e.g., only 2 in Fig.~\ref{fig:ElfXOR}(a). 
To this end, the state-of-the-art work {\em Elf}~\cite{li2023elf} erases the last few bits (i.e., setting them as zeros) of each value to increase the number of trailing zeros, thus achieving remarkable performance. As shown in Fig.~\ref{fig:ElfXOR}(b), {\em Elf} quickly finds a small value $\delta$ satisfying $0 < \delta < 10^{-\alpha}=0.01$ to erase the last 44 bits of $v_t$, i.e., $v'_t = v_t - \delta$, where $\alpha$ is the decimal place count of $v_t$ (Definition~\ref{def:dpc}). By XORing $v'_t$ with $v'_{t-1}$, we can obtain $xor'_t$ with as many as 44 trailing zeros. During decompression, since $\alpha = 2$, $v_t = v'_t + \delta = $ $xor'_t \oplus v'_{t-1} + \delta$ and $0 < \delta < 10^{-\alpha}$, we can safely infer $v_t = 3.16\stkout{40625} + 0.01 = 3.17$ (see Section~\ref{sec:pre} for more details).



To encode the XORed result $xor'_t$ containing both many leading zeros and trailing zeros, {\em Elf} proposes novel encoding strategies shown in Fig.~\ref{fig:ElfEncodingStrategy}. There are four cases. \textbf{Case~\uppercase\expandafter{\romannumeral1}}: If $xor'_t = 0$, {\em Elf} simply writes two bits of flag `01'. \textbf{Case~\uppercase\expandafter{\romannumeral2}}: If $xor'_t \neq 0$ and not $C$ and $center_t \leq 16$ ($C$ is ``$lead_t = lead_{t - 1}$ and $trail_t \geq trail_{t-1}$''; $lead_t$, $trail_t$ and $center_t$ are the numbers of leading zeros, trailing zeros and center bits of $xor'_t$, respectively), after writing two bits of flag `10', {\em Elf} writes three bits of a number to \underline{approximately} represent $lead_t$, four bits of $center_t$ and finally $center_t$ bits of center bits. \textbf{Case~\uppercase\expandafter{\romannumeral3}}: It is similar to Case~\uppercase\expandafter{\romannumeral2} but {\em Elf} writes six bits of $center_t$ since $center_t > 16$. \textbf{Case~\uppercase\expandafter{\romannumeral4}}: If $xor'_t \neq 0$ and the condition $C$ holds, {\em Elf} writes two bits of flag `00' and then writes the center bits directly, since we can approximate $lead_t$ and $trail_t$ by $lead_{t - 1}$ and $trail_{t-1}$, respectively.

\begin{figure}[t]
  \centering
  \includegraphics[width=3.5in]{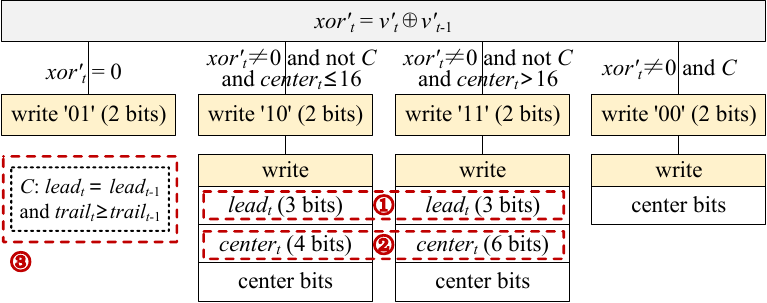}%
\vspace{-5pt}
  \caption{Encoding Strategies of {\em Elf}.}
  \label{fig:ElfEncodingStrategy}
  \vspace{-15pt}
\end{figure}

However, the encoding strategies of {\em Elf} still leave much room for improvement. First, {\em Elf} approximates $lead_t$ using 3 bits without considering the different leading zero distributions of different datasets (detailed in Section~\ref{sec:leading}). Second, including one extra flag bit for distinguishing between Case~\uppercase\expandafter{\romannumeral2} and Case~\uppercase\expandafter{\romannumeral3} in Fig.~\ref{fig:ElfEncodingStrategy}, {\em Elf} in fact records $center_t$ with 5 or 7 bits, which is too costly (detailed in Section~\ref{sec:centerbits}). Third, {\em Elf} employs condition $C$ to determine whether to share the information of leading zeros and trailing zeros with $xor'_{t-1}$, which may be suboptimal, because the saved bits may not offset the approximation cost (detailed in Section~\ref{sec:sharing}). 

To this end, this paper proposes {\em Elf}$^*$, a lossless compression algorithm for floating-point time series. Specifically, {\em Elf}$^*$ records $lead_t$ with an approximation rule (Definition~\ref{def:apprule}) that is adaptive to the distribution of leading zeros in a time series. Moreover, instead of encoding $center_t$, {\em Elf}$^*$ encodes $trail_t$ with another adaptive approximation rule based on the same technique proposed for $lead_t$, since with $lead_t$ and $trail_t$, we can calculate $center_t$ easily. Furthermore, {\em Elf}$^*$ employs an adaptive condition to determine whether or not to share the information with the previous $xor'_{t-1}$. Overall, the contributions of this paper are summarized as follows:

(1)~We give a further exploration on the encoding strategies of the state-of-the-art lossless floating-point compression {\em Elf}, and propose a set of optimizations for leading zeros, center bits and sharing condition. Our optimizations can enhance the compression ratio tremendously.

(2)~We develop a dynamic programming algorithm with a set of pruning strategies to compute the adaptive approximation rules efficiently. The approximation rules can be guaranteed to be globally optimal.

(3)~We extend {\em Elf}$^*$ to {\em SElf}$^*$ for streaming scenarios. {\em SElf}$^*$ can achieve almost the same compression ratio as {\em Elf}$^*$, while enjoying even higher efficiency.

(4)~We compare {\em Elf}$^*$ and {\em SElf}$^*$ with 8 state-of-the-art compression algorithms (including 5 streaming floating-point compression algorithms and 3 batch general compression algorithms) using 22 datasets. Among the streaming algorithms, {\em SElf}$^*$ achieves an average compression ratio improvement of 9.1\% over the best competitor {\em Elf}+, while maintaining a similar efficiency. Among the batch algorithms, when the block is small,  {\em Elf}$^*$ outperforms the best competitor Xz by 10.1\% in terms of compression ratio, but takes only 4.8\% compression time. When the block is big, {\em Elf}$^*$ still ranks among the most competitive compressors.

In the remainder of this paper, we provide some preliminaries in Section~\ref{sec:pre}. We describe the optimizations for leading zeros, center bits and sharing condition in Section~\ref{sec:leading}, Section~\ref{sec:centerbits} and Section~\ref{sec:sharing}, respectively. In Section~\ref{sec:stream}, we extend {\em Elf}$^*$ to {\em SElf}$^*$ for streaming scenarios. The experimental results are presented in Section~\ref{sec:exp}, followed by the related works in Section~\ref{sec:related}. We conclude this paper in Section~\ref{sec:conclude}.

\setlength{\tabcolsep}{0.1em} 
\begin{table}[t]
\caption{Symbols and Their Meanings}
\centering\label{tbl:symbols}
\begin{tabular}{|c|l|} 
\hline
\textbf{Symbols}			&\textbf{Meanings}	\\
\hline
\hline
\multirow{2}{*}{$TS = \langle v_{t_0}, v_{t_1}, ...v_{t_{n-1}} \rangle$} & Floating-point time series, where $v_{t_i}$ is a\\
& floating-point value generated at timestamp $t_i$\\
\hline
$v_t$, $v_{t-1}$	& Aliases of $v_{t_i}$ and $v_{t_{i-1}}$, respectively	\\
\hline
$\alpha = DP(v), \beta = DS(v)$	&	Decimal place count, decimal significand count\\
\hline
$XOR' = \langle xor'_{t_0}, xor'_{t_1},$ &Erased floating-point time series, where $xor'_0$ \\ 
$... xor'_{t_{n-1}}\rangle$				& $= v'_{t_0}$ and $xor'_i = v'_{t_i} \oplus v'_{t_{i-1}}$ for $i > 0$	\\
\hline
\multirow{2}{*}{$lead_t$, $trail_t$}	& Numbers of leading zeros and trailing zeros of \\
										& $xor'_t$, respectively\\
\hline
$XOR'_{lead_t = k}$ 					& $\{xor'_t \in XOR' $ $|$ $ lead_t = k\}$\\
\hline
$XOR'_{lead_t \leq k}$ 					& $\{xor'_t \in XOR' $ $|$ $ lead_t \leq k\}$\\
\hline
$rule = \langle a_0, a_1,... a_{z-1}\rangle$		& Approximation rule with a length of $z = |rule|$ \\
\hline
$cd = \langle c_0, c_1, ..., c_{63}\rangle$		& Leading zeros count distribution	\\
\hline
$C_{Total}$, $C_{App}$, $C_{Pre}$		& Total cost, approximation cost, presentation cost	\\
\hline
$nz$, $nz_f = \langle f_1, ..., f_{63} \rangle$,	& Non-zeros count, front non-zeros count array, \\
 $nz_r = \langle r_1, ..., r_{63} \rangle$	& rear non-zero count array of $\langle c_1, ..., c_{63}\rangle$ \\
\hline
\end{tabular}
\vspace{-10pt}
\end{table}

\section{Preliminaries}\label{sec:pre}

In this section, we give some preliminaries and definitions. Table~\ref{tbl:symbols} lists the symbols used frequently throughout this paper.

\begin{mydef}
\textbf{Floating-Point Time Series.} A floating-point time series (abbr. time series) is a time-ordered sequence of floating-point values $TS = \langle v_{t_0}, v_{t_1}, ... v_{t_{n-1}}\rangle$, where $v_{t_i}$ represents the value generated at time $t_i$. If there is no confusion, $v_{t_i}$ and its predecessor $v_{t_{i-1}}$ are denoted by $v_t$ and $v_{t-1}$, respectively. We also let $v_t$ or $v$ denote a floating-point value if we do not care about its position in a time series. $lead_t$ and $trail_t$ are the numbers of leading zeros and trailing zeros in the underlying storage of $v_t$, respectively.
\end{mydef}

The floating-point value $v$ can be of either a double-precision type (i.e., \textbf{double} value with 64 bits) or a single-precision type (i.e., \textbf{single} value with 32 bits). Like {\em Elf}~\cite{li2023elf} and Chimp~\cite{liakos2022chimp}, this paper mainly focuses on the compression for double values. However, our proposed algorithm can be easily transplanted for single values.

\begin{mydef}\label{def:dpc}
\textbf{Decimal Place Count $\alpha$ and Decimal Significand Count $\beta$~\cite{li2023elf}.} Given a floating-point value $v$, its decimal format is $DF(v) = \pm(d_{h-1}d_{h-2}...d_0.d_{-1}d_{-2}...d_{l})_{10}$, where $d_i \in \{0, 1, ..., 9\}$ for $l \leq i \leq h-1$, $d_{h - 1} \neq 0$ unless $h - 1 = 0$, and $d_l \neq 0$ unless $l = -1$. The decimal place count $\alpha$ is calculated by $DP(v) = |l|$. If for all $l < n \leq i \leq h - 1$, $d_i = 0$ but $d_{n - 1} \neq 0$, the decimal significand count $\beta$ is $DS(v) = n - l$. In the case of $v = 0$, we let $DS(v) = 0$.
\end{mydef}

For example, $DP(3.17) = 2$ and $DS(3.17) = 3$; $DP(-0.0317) = 4$ and $DS(-0.0317) = 3$. 

\begin{mydef}
\textbf{Erased Floating-Point Time Series.} Given $TS = \langle v_{t_0}, v_{t_1}, ... v_{t_{n-1}}\rangle$, its corresponding erased floating-point time series $TS' = \langle v'_{t_0}, v'_{t_1}, ... v'_{t_{n-1}}\rangle$ is obtained by:
\begin{equation}\label{equ:erasing}
v'_{t_i} = \left\{  
	\begin{array}{ll}
		0,	&	\text{if }v_{t_i} = 0\\
		Erase(v_{t_i}),	&	\text{if }v_{t_i} \neq 0
	\end{array}
\right.
\end{equation}

\noindent where $Erase(v_{t_i})$ is an erasing operation that sets the last $52 - (\lceil DP(v_{t_i}) \times log_2{10} \rceil + \lfloor log_2{|v_{t_i}|} \rfloor)$ bits to be zeros in the underlying storage of $v_{t_i}$~\cite{li2023elf}.
\end{mydef}

For example, as shown in Fig.~\ref{fig:ElfXOR}(b), by erasing the last $52 - (\lceil 2 \times log_2{10} \rceil + \lfloor log_2{|3.17|} \rfloor) = 44$ bits of $v_t = 3.17$, we can get $v'_t = Erase(3.17) = 3.1640625$.

Each value $v'_{t_i} \in TS'$ is supposed to contain much more trailing zeros than $v_{t_i} \in TS$. The work~\cite{li2023elf} proves that the erasing operation can erase as many as $x$ bits of $v_{t_i}$, where $51 - \beta log_2{10} < x < 53 - (\beta - 1) log_2{10}$. Here, $\beta = DS(v_{t_i})$ is the decimal significand count of $v_{t_i}$. To recover $v_{t_i}$ from $v'_{t_i}$, {\em Elf} performs the following restoring operation:
\begin{equation}\label{equ:restoring}
v_{t_i} = \left\{  
	\begin{array}{ll}
		0,	&	\text{if }v'_{t_i} = 0\\
		LeaveOut(v'_{t_i}, \alpha) + 10^{-\alpha},	&	\text{if }v'_{t_i} \neq 0
	\end{array}
\right.
\end{equation}

\noindent where $LeaveOut(v'_{t_i}, \alpha) = \pm(d_{h'-1}d_{h'-2}...d_0.d_{-1}...d_{-\alpha}\\\stkout{d_{-(\alpha + 1)}...d_{l'}})_{10}$ leaves out the digits after $d_{-\alpha}$ in $DF(v'_{t_i})$.

For the example shown in Fig.~\ref{fig:ElfXOR}, $v_t = LeaveOut(v'_t, 2) + 10^{-2} = 3.16\stkout{40625} + 0.01 = 3.17$.

\begin{mydef}
\textbf{Erased XORed Time Series.} Given $TS' = \langle v'_{t_0}, v'_{t_1}, ... v'_{t_{n-1}}\rangle$, its corresponding erased XORed time series $XOR' = \langle xor'_{t_0}, xor'_{t_1}, ... xor'_{t_{n-1}}\rangle$ is calculated by:
\begin{equation}\label{equ:xoring}
xor'_{t_i} = \left\{  
	\begin{array}{ll}
		v_{t_0},	&	\text{if }i = 0\\
		v'_{t_{i-1}} \oplus v'_{t_i},	&	\text{if } i > 0
	\end{array}
\right.
\end{equation}
\end{mydef}

\noindent where $\oplus$ is the bitwise XORing operation. Since both $v'_{t_{i-1}}$ and $v'_{t_i}$ are supposed to contain many trailing zeros, $xor_{t_i}$ is expected to have many trailing zeros as well. We let $XOR'_{lead_t = k} = \{xor'_t \in XOR' $ $|$ $ lead_t = k\}$ and $XOR'_{lead_t \leq k} = \{xor'_t \in XOR' $ $|$ $ lead_t \leq k\}$.

\section{Optimization for Leading Zeros}\label{sec:leading}

In this section, we first introduce the existing solutions for encoding leading zeros, and then propose an adaptive leading zeros encoding strategy with two main steps: locally best rule calculation and globally best rule calculation.

\subsection{Existing Solutions for Encoding Leading Zeros}

\begin{figure}[t]
  \centering
  \includegraphics[width=3.5in]{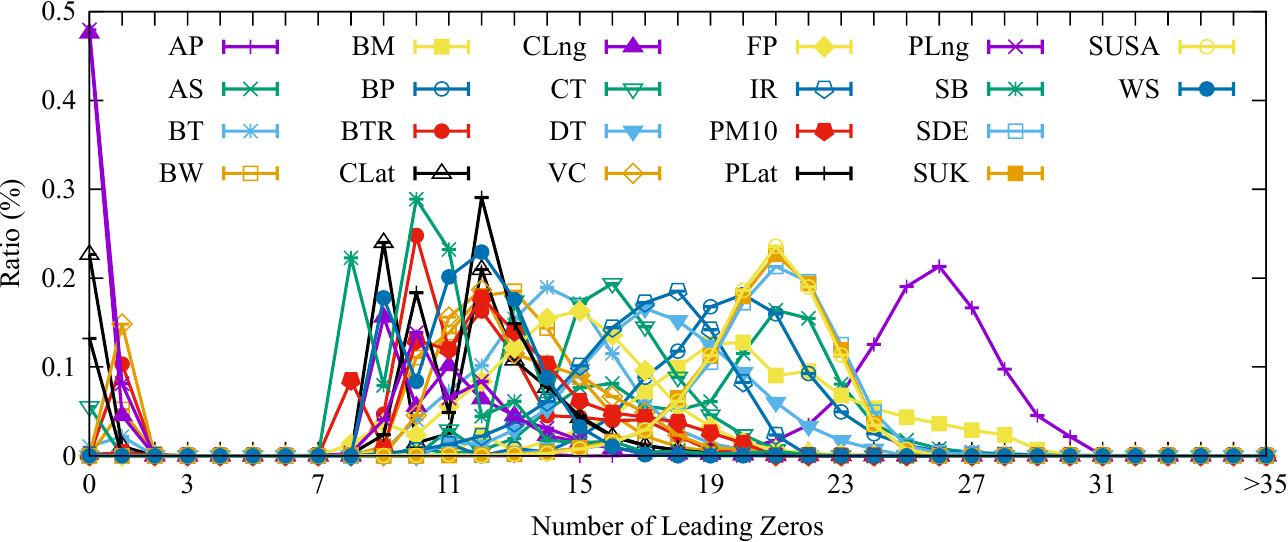}%
\vspace{-5pt}
  \caption{Distribution of Leading Zeros Number.}
  \label{fig:LeadingZerosCountDistribution}
  \vspace{-10pt}
\end{figure}

To encode the number of leading zeros $lead_t$, one basic solution is to leverage $\lceil log_2{64}\rceil = 6$ bits, since there could be 64 different values of $lead_t$ if $xor'_t \neq 0$. Note that we ignore the case when $xor'_{t} = 0$, because we will handle it specially as shown in Fig.~\ref{fig:ElfEncodingStrategy}. Figure~\ref{fig:LeadingZerosCountDistribution} shows the leading zeros distributions of various datasets (please refer to Section~\ref{sec:exp} for more details about these datasets). Considering $lead_t$ is rarely bigger than 31, Gorilla~\cite{pelkonen2015gorilla} proposes the first \textbf{approximation rule} that utilizes 5 bits to encode $lead_t$. Specifically, if $lead_t < 31$, Gorilla stores the real value of $lead_t$ with 5 bits. If $lead_t \geq 31$, Gorilla approximates it to 31, with an \textbf{approximation cost} of $lead_t - 31$. We here give the formal definitions of approximation rule, approximation cost and presentation cost.

\begin{mydef}\label{def:apprule}
\textbf{Approximation Rule, Approximation Cost and Presentation Cost.} An approximation rule (abbr. rule) is an ordered integer items array $rule = \langle a_0, a_1, ..., a_{z-1}\rangle$, where $0 \leq a_i < a_j \leq 63$ for all $0 \leq i < j < z$. Given a number of leading zeros $lead_t$, if $\exists i \in [0, z - 1)$, $a_i \leq lead_t < a_{i + 1}$, then $lead_t$ is approximated to $a_i$, denoted as $App(lead_t)=a_i$. If $lead_t \geq a_{z-1}$, then $lead_t$ is considered to be $a_{z-1}$, denoted as $App(lead_t) = a_{z-1}$. We call $App(lead_t)$ the approximated number of leading zeros, and $lead_t - App(lead_t)$ the approximation cost of $lead_t$ (as illustrated in Fig.~\ref{fig:AppRule}(a)). To represent an approximated number, it requires at least $\lceil log_2{z} \rceil$ bits, where $z = |rule|$. We call $\lceil log_2{z} \rceil$ the presentation cost.
\end{mydef}

\begin{figure}[htb]
  \centering
  \vspace{-10pt}
  \includegraphics[width=3.5in]{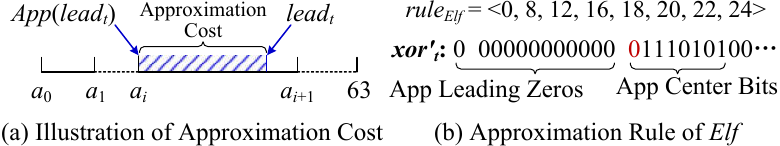}%
\vspace{-5pt}
  \caption{Illustration of Approximation Rule and Cost.}
  \label{fig:AppRule}
\end{figure}

For example, the rule proposed by Gorilla described above can be deemed as $\langle 0, 1, 2, ..., 31\rangle$ with a presentation cost of $\lceil log_2{32} \rceil = 5$. Figure~\ref{fig:AppRule}(b) presents the rule adopted by {\em Elf}, which has a presentation cost of $\lceil log_2{8} \rceil = 3$~\footnote{In this case, we map the value of 3 bits to each $a_i$ one-by-one, e.g., 0 to 0, 1 to 8, 2 to 12, and so on.}. As shown in Fig.~\ref{fig:AppRule}(b), given an $xor'_t$ with $lead_t = 13$, {\em Elf} will treat it as $App(lead_t) = 12$, and regard the 13$^{\text{th}}$ zero as a part of center bits. As a result, it would take one more bit to store the center bits, which is where the approximation cost comes from.

Because this paper aims to compress floating-point values losslessly, we should always set $a_0 = 0$. Otherwise, if $a_0 \neq 0$, we cannot deal with the cases when $lead_t < a_0$, especially for the future unseen values in a time series. 

\begin{mydef}
\textbf{Total Cost.} Given an erased XORed time series $XOR' = \langle xor'_{t_0}, xor'_{t_1}, ... xor'_{t_{n-1}}\rangle$ and its leading zeros count distribution $cd = \langle c_0, c_1, ..., c_{63}\rangle$, where $c_i = |XOR'_{lead_t = i}|$, the total cost of $rule = \langle a_0, a_1, ..., a_{z-1}\rangle$ on $XOR'$ is:
\begin{equation}\label{equ:totalcost2}
\begin{aligned}
	C_{Total} &= C_{App} + C_{Pre}\\
	  &= \sum\limits_{i=0}^{63}c_i \times (i - a_j) + \sum\limits_{i=0}^{63}c_i \times \lceil log_2{z} \rceil, \\
	  &\text{subject to: } 0 \leq a_j \leq i < a_{j+1} \text{ and } a_z = 64
\end{aligned}
\end{equation}

\noindent $C_{App} = \sum_{i=0}^{63}c_i \times (i - a_j)$ is the total approximation cost and $C_{Pre} = \sum_{i=0}^{63}c_i \times \lceil log_2{z} \rceil$ is the total presentation cost.
\end{mydef}

The number of rule items $z = |rule|$ plays a trade-off between $C_{App}$ and $C_{Pre}$. A bigger $z$ usually means a smaller $C_{App}$ but a larger $C_{Pre}$. Moreover, the value of each $a_i \in rule$ has a large impact on the total approximation cost. As shown in Fig.~\ref{fig:LeadingZerosCountDistribution}, different datasets have different distributions of leading zeros. Therefore, the fixed rule specified by {\em Elf} for all datasets is generally suboptimal.

\subsection{Adaptive Leading Zeros Encoding Strategy}


To find a rule that yields the minimum total cost, one na\"ive method is to enumerate all the possible rule item combinations, calculate the total cost of each combination, and find the one with the minimum total cost. However, this method is prohibitively expensive, since there are as many as $2^{64}$ possible rules. To this end, this paper formulates the best rule search problem as a dynamic programming process. 

In what follows, we first develop a method to find the best rule with exact $z$ rule items (i.e., locally best rule). Using this as a building block, we can find the globally best rule.

\subsubsection{Locally Best Rule Calculation} We describe the locally best rule calculation method from the following aspects.

\textbf{States Calculation.} Since we fix the total rule items in the locally best rule to be $z$, we consider only the total approximation cost $C_{App}$ but ignore the total presentation cost $C_{Pre}$, because all rules with the same number of rule items have the same $C_{Pre}$. We resort to a two-dimensional array $dp[*][*]$ to record the calculation state, where $d[i][j]$ denotes the total approximation cost when we set $a_j = i$ and only consider $xor'_t \in XOR'_{lead_t \leq i}$. As a consequence, we have the following state transition equation:
\begin{equation}\label{equ:statetransition}
\footnotesize
dp[i][j] = \left\{  
\begin{array}{ll}
	0,	 \text{\qquad\qquad\quad\;\;\, if } i = 0 \text{ and } j = 0 \\
	\sum\limits_{p = 0}^{i - 1} c_p \times p,	 \text{\qquad if } i > 0 \text{ and } j = 1\\
	\infty,	 \text{\qquad\qquad\quad\; if } i < j \\
	\min\limits_{j-1 \leq k \leq i-1}(dp[k][j-1] + \sum\limits_{p=k+1}^{i-1}c_p \times (p-k)),  \text{ others}
\end{array}
\right.
\end{equation}

We can understand Equation~(\ref{equ:statetransition}) with four cases.
\textbf{Case~\uppercase\expandafter{\romannumeral1}}: If $i = 0$ and $j = 0$, we set $a_0 = 0$. For all $xor'_t \in XOR'_{lead_t \leq 0}$, we have $lead_t = 0$ and $App(lead_t) = 0$, i.e., its approximation cost is 0, so $dp[0][0] = 0$. 
\textbf{Case~\uppercase\expandafter{\romannumeral2}}: If $i > 0$ and $j = 1$, we set $a_1 = i$. For all $xor'_t \in XOR'_{lead_t = i}$, we have $App(lead_t) = i$ with an approximation cost of 0. For each $xor'_t \in XOR'_{lead_t \leq i - 1}$, since $a_0$ must be set to 0, we have $App(lead_t) = a_0 = 0$ with an approximation cost of $lead_t$. As $c_p = |XOR'_{lead_t = p}|$, hence $dp[i][j] = \sum_{p=0}^{i-1}c_p \times p$.
\textbf{Case~\uppercase\expandafter{\romannumeral3}}: If $i < j$, we let $dp[i][j] = \infty$, since $a_j$ must be greater than or equal to $j$ in accordance with the definition of approximation rule (i.e., $a_x < a_y$ for all $0 \leq x < y$).
\textbf{Case~\uppercase\expandafter{\romannumeral4}}: If $i \geq j > 1$, we should try all cases of $k \in [j - 1, i - 1]$, and find the one that makes $dp[k][j - 1] + \sum_{p = k + 1} ^ {i - 1} c_p \times (p - k)$ minimized. The former (i.e., $dp[k][j - 1]$) is the state when we set $a_{j - 1} = k$ (therefore it requires $k \geq j - 1$), and the latter (i.e., $\sum_{p=k+1}^{i-1}c_p \times (p - k)$) is the summation of approximation costs of those leading zeros that are approximated to $k$. Since we make fully use of the intermediate calculation states, we can avoid a large number of repetitive computations.

With $dp[*][*]$, we can find the minimum total approximation cost of a rule whose number of items is $z$, i.e., 
\begin{equation}\label{equ:localAppCost}
C_{App} = \min\limits_{z-1 \leq k \leq 63}dp[k][z-1] + \sum_{p = k + 1}^{63}c_p\times(p - k)
\end{equation}

\noindent which is similar to Case~\uppercase\expandafter{\romannumeral4} in Equation~(\ref{equ:statetransition}) but we set $i = 64$ and $j = z$. We let $a_{z - 1} = k^*$, where  $k^* \in [z - 1, 63]$ is the value of $k$ that minimizes Equation~(\ref{equ:localAppCost}).

\textbf{Rule Items Calculation.} To obtain the items of the locally best rule based on $dp[*][*]$, we leverage another two-dimensional array $pre[*][*]$. Specifically, for Case~\uppercase\expandafter{\romannumeral1}, we set $pre[0][0] = -1$, which indicates that $a_0 = 0$ is the first rule item; for Case~\uppercase\expandafter{\romannumeral2}, we set $pre[i][1] = 0$ where $i > 0$, which means that the rule item just before $a_1 = i$  is equal to 0 (i.e., $a_0 = 0$); we ignore the Case~\uppercase\expandafter{\romannumeral3} for $pre[i][j]$ where $i < j$ because it certainly will not be the final result; for Case~\uppercase\expandafter{\romannumeral4}, we set $pre[i][j] = k^*_1$, where $k^*_1 \in [j-1, i - 1]$ is the value of $k$ that minimizes $dp[k][j - 1] + \sum_{p = k + 1}^{i-1}c_p \times (p - k)$. Using $k^*$ (for minimizing Equation~(\ref{equ:localAppCost})) and $pre[*][*]$, we can calculate $a_i$ in the locally best rule backwards (detailed in Algorithm~\ref{alg:localAppRule}).

\textbf{Pruning Strategies.} It would encounter efficiency issue for the calculation of $dp[*][*]$, since it requires four nested loops to implement Case~\uppercase\expandafter{\romannumeral4} in Equation~(\ref{equ:statetransition}) (i.e., $i$ from 0 to 63, $j$ from 0 to $z - 1$, $k$ from $j - 1$ to $i - 1$, and $p$ from $k + 1$ to $i - 1$), which results in a quartic time complexity. To resolve this issue, this paper proposes a set of pruning strategies. 
Essentially, a rule is a one-to-one mapping function (a.k.a. injective function) from $j \in [0, z - 1]$ to $i \in [0, 63]$ (i.e., $j\stackrel{a}{\longrightarrow}i$). 
Since $a_0$ should always be $0$, in the following, we consider only the mapping from $j \in [1, z - 1]$ to $i \in [1, 63]$. Our intuition is simple: for any $i' \in [0, 63]$, if $c_{i'} = 0$, then $i'$ should not be mapped by any $j$. The intuition is formalized as Theorem~\ref{theorem:prune1}.

\begin{Theorem}\label{theorem:prune1}
Given $cd = \langle c_0, c_1, ..., c_{63}\rangle$, let $nz = |\{c_i \in \langle c_1, ..., c_{63}\rangle $ $|$ $ c_i \neq 0\}|$, $rule = \langle 0, a_1, ..., a_{z-1}\rangle$ is the locally best rule with a length of $z$, if $z-1 \leq nz$, then for any $j \in [1, z - 1]$, $c_{a_j} \neq 0$.
\end{Theorem}
\begin{proof}
It is easy to understand that if $a_j = i$, for all $xor'_t \in XOR'_{lead_t = i}$, their approximation costs can reach the minimum value (i.e., $0$). 

Suppose there exists one $j' \in [1, z - 1]$ such that $a_{j'} = i'_0$ and $c_{i'_0} = 0$. Next, we map each $j \in [1, j') \cup (j', z - 1]$ to $i \in [1, 63]$ (i.e., map $z - 2$ values to $[1, 63]$, where there are still $nz$ unmapped non-zeros in $\langle c_1, ..., c_{63}\rangle$). Because $z - 1 \leq nz$, we have $z - 2 < nz$. Therefore, there must exist at least one $i^* \in [1, 63]$, where $c_{i^*} \neq 0$ but $i^*$ is not mapped. There are two cases: 1) there exists $i^* > i'_0$; 2) there does not exist $i^* > i'_0$ so $i^* < i'_0$. For the former case, we let $i^* = i'_{k +1}$ is the smallest one that $c_{i'_{k+1}} \neq 0$ but without mapping from $j \in [1, z - 1]$. As shown in Fig.~\ref{fig:Prune1}(a), we can remap $\langle j', j'+1, ..., j'+k\rangle$ to $\langle i'_1, i'_2, ..., i'_{k+1}\rangle$, which reduces the approximation costs after $i'_{0}$ but does not affect the approximation costs before $i'_{0}$, i.e., it reduces the total approximation cost. For the latter case, we let $i'_{-k-1} < i'_0$ is the biggest one that $c_{i'_{-k-1}} \neq 0$ but without mapping from $j \in [1, z - 1]$. As shown in Fig.~\ref{fig:Prune1}(b), we can remap $\langle j'-k, j'-k + 1, ..., j'\rangle$ to $\langle i'_{-k - 1}, i'_{-k}, ..., i'_{-1} \rangle$, which reduces the approximation costs before $i'_0$ but does not affect the approximation costs after $i'_0$, i.e., it reduces the total approximation cost as well. That is, for both cases, we can find another rule whose total approximation cost is smaller than the locally best one, which contradicts our hypothesis.
\end{proof}
\begin{figure}[htb]
  \centering
  \vspace{-10pt}
  \includegraphics[width=3.5in]{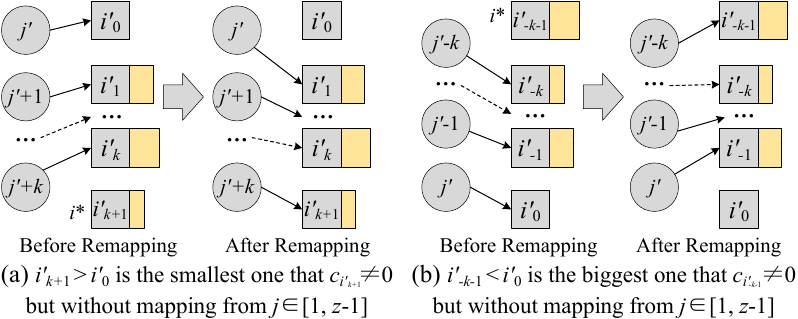}%
\vspace{-5pt}
  \caption{Demonstration for Proof of Theorem~\ref{theorem:prune1}.}
  \label{fig:Prune1}
\end{figure}

Theorem~\ref{theorem:prune1} can act as a pruning rule (i.e., \textbf{Zero Pruning}) for the calculation of $dp[*][*]$. We can also leverage the number of non-zeros in $cd = \langle c_0, c_1, ..., c_{63}\rangle$ to accelerate the calculation process, which is formalized as Theorem~\ref{theorem:prune2}.

\begin{Theorem}\label{theorem:prune2}
Given $cd = \langle c_0, c_1, ..., c_{63}\rangle$, let $nz = |\{c_i \in \langle c_1, ..., c_{63}\rangle $ $|$ $ c_i \neq 0\}|$, $nz_{f} = \langle f_1, ..., f_{63}\rangle$ where $f_j = |\{c_i \in \langle c_1, ..., c_{63}\rangle $ $|$ $ c_i \leq j $ and $ c_i \neq 0\}|$, $nz_{r} = \langle r_1, ..., r_{63}\rangle$ where $r_j = |\{c_i \in \langle c_1, ..., c_{63}\rangle $ $|$ $ c_i > j $ and $ c_i \neq 0\}|$, $rule = \langle 0, a_1, ..., a_{z-1}\rangle$ is the locally best rule with a length of $z$ where $z - 1 \leq nz$, for any $i' \in [1, 63]$ and $j' \in [1, z - 1]$, if $f_{i'} < j'$ or $r_{i'} < z - 1 - j'$, then $a_{j'} = i'$ will not be in $rule$.
\end{Theorem}
\begin{proof}
According to Theorem~\ref{theorem:prune1}, if $z - 1 \leq nz$, each $j \in [1, z - 1]$ must be mapped into a distinct $i \in [1, 63]$ where $c_i \neq 0$. $f_{i'}$ records the number of non-zeros in $\langle c_1, c_2, ..., c_{i'} \rangle$. If $f_{i'} < j'$, it means that there are not enough $i \in [1, i']$ where $c_i \neq 0$ for the mapping of all $j \in [1, j']$. Similarly, $r_{i'}$ records the number of non-zeros in $\langle c_{i'+1}, c_{i' + 2}, ..., c_{63} \rangle$. If $r_{i'} < z - 1 - j'$, it means that there are not enough $i \in [i'+1, 63]$ where $c_i \neq 0$ for the mapping of all $j \in [j' + 1, z - 1]$. Therefore, if $f_{i'} < j'$ or $r_{i'} < z - 1 - j'$, $a_{j'} = i'$ will not in the final result.
\end{proof}

Theorem~\ref{theorem:prune2} provides another pruning rule (i.e., \textbf{Front-Rear Pruning}) for the calculation of $dp[*][*]$. Note that for both Theorem~\ref{theorem:prune1} and Theorem~\ref{theorem:prune2}, it requires $z - 1 \leq nz$. To handle the case of $z - 1 > nz$, we have Theorem~\ref{theorem:ZNZ}.

\begin{Theorem}\label{theorem:ZNZ}
Given $cd = \langle c_0, c_1, ..., c_{63}\rangle$, let $nz = |\{c_i \in \langle c_1, ..., c_{63}\rangle $ $|$ $ c_i \neq 0\}|$, $z$ is the target number of rule items. For $rule = \langle 0, a_1, ..., a_{z-1}\rangle$ where $z - 1 \geq nz$, if for any $i \in [1, 63]$ that $c_i \neq 0$, there exists one $j \in [1, z - 1]$ such that $a_j = i$, then $rule$ is the locally best one, no matter what the actual value of $a_{j'}$ is for those $j'$ satisfying $c_{a_{j'}} = 0$.
\end{Theorem}
\begin{proof}
The total approximation cost is $\sum_{i = 0}^{63}c_i \times (i - a_j)$. For those $c_i \neq 0$, we have $a_j = i$, i.e., $c_i \times (i - a_j) = 0$; for those $c_i = 0$, we have $c_i \times (i - a_j) = 0$ no matter what the actual value of $a_j$ is. To this end, $rule$ is the locally best rule with an approximation cost of 0.
\end{proof}

Theorem~\ref{theorem:ZNZ} indicates that, to get the locally best rule with a length of $z$, if $z - 1 > nz$, we can first find a rule  $rule$ with the length of $nz + 1$, and then adjust $rule$ to make $|rule| = z$. In this paper, we assign the smallest unmapped value $i$ that satisfies $c_i = 0$ to $a_j$. Other more complex assignment strategies are left as open problems.

\textbf{Locally Best Rule Algorithm.} Algorithm~\ref{alg:localAppRule} depicts the pseudo-code of locally best rule calculation, which consists of three main parts: \textit{Initialization}, \textit{States Calculation} and \textit{Rule Items Calculation}.

\begin{algorithm}[t]
\small
	\caption{$LocalAppRule(cd = \langle c_0, c_1, ..., c_{63}\rangle, z)$}\label{alg:localAppRule}
	\tcc{Initialization}
	$nz \leftarrow |\{c_i \in \langle c_1, ..., c_{63}\rangle $ $|$ $ c_i \neq 0\}|$\label{alg:localAppRule:init:start}\;
	$nz_{f} \leftarrow \langle f_1, ..., f_{63}\rangle$, where $f_j = |\{c_i \in \langle c_1, ..., c_{63}\rangle $ $|$ $ c_i \leq j $ and $ c_i \neq 0\}|$\;
	$nz_{r} \leftarrow \langle r_1, ..., r_{63}\rangle$, where $r_j = |\{c_i \in \langle c_1, ..., c_{63}\rangle $ $|$ $ c_i > j $ and $ c_i \neq 0\}|$\;
	$z' \leftarrow min(z, nz + 1)$\label{alg:localAppRule:init:end}\;
	\tcc{States Calculation}
	$dp[0][0] \leftarrow 0$; $pre[0][0] \leftarrow -1$\label{alg:localAppRule:dp:start}; \tcp{Case~\uppercase\expandafter{\romannumeral1}}
	\For{$i \leftarrow 1 \text{\rm\ to } 63$}{
		\If{$c_i = 0$}{
			\Continue;\label{alg:localAppRule:prune1}	\tcp{Zero Pruning}
		}
		\For{$j \leftarrow 1 \text{\rm\ to } min(i, z'-1)$\label{alg:localAppRule:case3}}{
			\If(\tcp*[h]{Case~\uppercase\expandafter{\romannumeral2}}){$i > 0 {\rm\ and\ } j = 1$\label{alg:localAppRule:dp:case2:start}}{
				$dp[i][j] \leftarrow \sum_{p = 0}^{i - 1}c_p \times p$; $pre[i][j]=0$\;\label{alg:localAppRule:dp:case2:end}
			}
			\Else(\tcp*[h]{Case~\uppercase\expandafter{\romannumeral4}}\label{alg:localAppRule:case4:start}){
				\If{$f_i < j {\rm\ or\ } r_i < z' - 1 - j$}{
					\Continue;\label{alg:localAppRule:prune2}\tcp{Front-Rear Pruning}			
				}
				$dp[i][j] \leftarrow \min\limits_{j-1 \leq k \leq i-1}(dp[k][j-1] + \sum_{p=k+1}^{i-1}c_p \times (p-k))$\label{alg:localAppRule:dp1}\;
				$pre[i][j] \leftarrow k^*_1$\label{alg:localAppRule:dp:end};\tcp{$k^*_1$ minimizes Line~\ref{alg:localAppRule:dp1}}
			}
		}		
	}
	\tcc{Rule Items Calculation}
	$C_{App} \leftarrow \min\limits_{z'-1 \leq k \leq 63}dp[k][z'-1] + \sum_{p = k + 1}^{63}c_p\times(p - k)$\label{alg:localAppRule:ar:start}\;
	Set $k^*$ as the value of $k$ that minimizes Line~\ref{alg:localAppRule:ar:start}\;
	$rule = \langle \ \rangle$; $i = 1$\;
	\While{$k^* \neq -1$}{
		$rule.prepend(k^*)$; $k^* \leftarrow pre[k^*][z' - i]$; $i++$\;
	}
	Adjust $rule$ to make $|rule|=z$\;
	\Return{$(rule, C_{App})$}\;\label{alg:localAppRule:ar:end}
\end{algorithm}

\textit{Initialization} (Lines~\ref{alg:localAppRule:init:start}-\ref{alg:localAppRule:init:end}). We first initialize $nz$, $nz_f$ and $nz_r$ based on the leading zeros count distribution $cd = \langle c_0, c_1, ..., c_{63} \rangle$. In Line~\ref{alg:localAppRule:init:end}, we assign the smaller value between $z$ and $nz + 1$ to $z'$, which means we will find a rule with a length of $nz + 1$ first if $z - 1 > nz$ according to Theorem~\ref{theorem:ZNZ}.

\textit{States Calculation} (Lines~\ref{alg:localAppRule:dp:start}-\ref{alg:localAppRule:dp:end}). In this part, we calculate the states in three cases, i.e., Case~\uppercase\expandafter{\romannumeral1} in Line~\ref{alg:localAppRule:dp:start}, Case~\uppercase\expandafter{\romannumeral2} in Lines~\ref{alg:localAppRule:dp:case2:start}-\ref{alg:localAppRule:dp:case2:end} and Case~\uppercase\expandafter{\romannumeral4} in Lines~\ref{alg:localAppRule:case4:start}-\ref{alg:localAppRule:dp:end}. Note that Case~\uppercase\expandafter{\romannumeral3} is considered implicitly, since we enforce that $j$ is not larger than $min(i, z'-1)$ in Line~\ref{alg:localAppRule:case3}. During the calculation of $dp[*][*]$, we apply the two pruning strategies, i.e., Zero Pruning and Front-Rear Pruning, in Line~\ref{alg:localAppRule:prune1} and Line~\ref{alg:localAppRule:prune2}, respectively.

\textit{Rule Items Calculation} (Lines~\ref{alg:localAppRule:ar:start}-\ref{alg:localAppRule:ar:end}). In this part, we first get the minimum total approximation cost $C_{App}$, and then calculate the rule items backwards with the help of $pre[*][*]$. After that, we adjust $rule$ by making $|rule| = z$ to meet the requirement of rule item number. Finally, the locally best rule and its total approximation cost are returned.

\textbf{Complexity Analysis.} In Algorithm~\ref{alg:localAppRule}, two two-dimensional arrays (i.e., $dp[*][*]$ and $pre[*][*]$) are utilized (but both are only half used). Therefore, the \textbf{space complexity} of Algorithm~\ref{alg:localAppRule} is $\mathcal{O}(64\times min(z, nz + 1)) \approx \mathcal{O}(64\times nz)$. The most time-consuming part is states calculation, in which there are four nested loops, i.e., $i$ from 1 to 63, $j$ from 1 to $min(i, z' - 1)$, $k$ from $j - 1$ to $i - 1$, and $p$ from $k + 1$ to $i - 1$. Thanks to Zero Pruning, $i$ loops at most $nz$ times actually. We can also apply the two pruning strategies when looping $k$.  Hence, the \textbf{time complexity} of Algorithm~\ref{alg:localAppRule} is $\mathcal{O}(nz \times min(z, nz + 1) \times nz \times 64) \approx \mathcal{O}(64\times nz^3)$.

To better illustrate Algorithm~\ref{alg:localAppRule}, we give an example as shown in Fig.~\ref{fig:LocalAppRuleExample}, in which $|cd|$ is set as 5 for simplification. Above the table, the values of $cd$, $z$, $nz$, $z'$, $nz_f$ and $nz_r$ are displayed, respectively. In the table, we show the calculation process. Here, $i$ loops from 0 to $|cd| - 1 = 4$, while $j$ loops from 1 to $min(i, z' - 1 = 2)$. As $c_1 = c_2 = 0$, we perform Zero Pruning when $i = 1$ or $2$. When $i = 3$ and $j = 2$, since $f_3 < 2$, we perform Front-Rear Pruning. If $i = 4$ and $j = 1$, because $r_4 < z' - 1 - 1$, we perform another Front-Rear Pruning. Subsequently, $k^*$ is assigned with 4 since it achieves the minimum approximation cost $C_{App} = 0$. After that, $rule = \langle 0, 3, 4 \rangle$ can be obtained with the help of $pre[*][*]$ (see the arrows in Fig.~\ref{fig:LocalAppRuleExample}). Because $|rule| < z$, we should adjust it by adding one value into it. Although there are two unmapped values of $i$ (i.e., 1 and 2), we prefer the smaller one (i.e., 1). Eventually, we have $rule = \langle 0, 1, 3, 4 \rangle$.

\begin{figure}[t]
  \centering
  \includegraphics[width=3.4in]{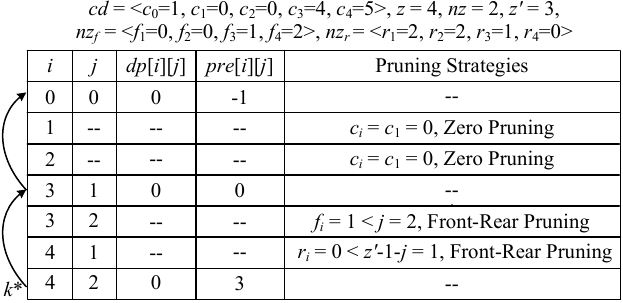}%
\vspace{-5pt}
  \caption{Example of Locally Best Rule.}
  \label{fig:LocalAppRuleExample}
  \vspace{-10pt}
\end{figure}

\subsubsection{Globally Best Rule Calculation} To obtain the globally best rule, one basic solution is to enumerate all possible values of $z$. For each value of $z$, we call Algorithm~\ref{alg:localAppRule} to get the locally best rule. Among these locally best rules, we select the one with the minimum total cost. Since there are as many as 64 possible values of $z$ (i.e., from 0 to 63), the basic solution is rather time-consuming.

To accelerate the process of globally best rule calculation, this paper proposes a set of optimizations.

\textbf{Opt~\uppercase\expandafter{\romannumeral1}: Exponential Values Only.} Recall that in Equation~(\ref{equ:totalcost2}), for $rule = \langle a_0, a_1, ..., a_{z - 1}\rangle$ with a length of $z$, its total cost is composed of two parts: the total approximation cost $C_{App} = \sum_{i = 0} ^ {63} c_i \times (i - a_j)$ and the total presentation cost $C_{Pre} = \sum_{i = 0} ^ {63} c_i \times \lceil log_2{z} \rceil$. Given another rule $rule' = \langle a_0, a_1, ...a', ..., a_{z - 1}\rangle$ with a length of $z + 1$, if $\lceil log_2{z} \rceil = \lceil log_2{(z + 1)} \rceil$, then $rule'$ has the same total presentation cost with $rule$. Besides, adding $a'$ to $rule$ will not increase the total approximation cost since we fix $a_j$ for all $j \in [0, z - 1]$. To this end, the total cost of $rule'$ is not higher than $rule$. Based on this observation, we can only enumerate $z \in \{2^0, 2^1, 2^2, 2^3, 2^4, 2^5, 2^6\}$, since for each $rule$ with a length of $z \in (2^k, 2^{k + 1})$, $0 \leq k \leq 5$, we can construct another $rule'$ with a length of $2^{k + 1}$, such that the total cost of $rule'$ is no more than that of $rule$.

\textbf{Opt~\uppercase\expandafter{\romannumeral2}: Skipping Some Exponential Values.} Theorem~\ref{theorem:ZNZ} demonstrates that, if $z - 1 \geq nz$, we can always find a $rule$ with the minimum approximation cost of 0. Combining with Opt~\uppercase\expandafter{\romannumeral1}, we can set the maximum value of $z$ as $2^{\lceil log_2{(nz + 1)}\rceil}$, because once $z$ reaches $2^{\lceil log_2{(nz + 1)}\rceil}$, increasing $z$ does not reduce the total approximation cost (it is already 0), but rather increases the total presentation cost. Furthermore, we find that for any $cd = \langle c_0, c_1, ..., c_{63}\rangle$, the total cost of the locally best rule with a length of $2^5$ is always no more than the one with a length of $2^6$. If $z = 2^6$, the locally best rule is $rule_1 = \langle 0, 1, ..., 63 \rangle$, whose total cost is $C_{Total1} = 0 + \sum_{i = 0}^{63} c_i \times 6$. If $z = 2^5$, we can construct a $rule_2 = \langle 0, 2, ..., 62 \rangle$, whose total cost is $C_{Total2} = c_1 + c_3 + ... + c_{63} + \sum_{i=0}^{63} c_i \times 5$. Here, we have $C_{Total1} - C_{Total2} = c_2 + c_4 + ... + c_{62} \geq 0$. Overall, we consider only $z \leq min(2^{\lceil log_2{(nz + 1)}\rceil}, 2^5)$.

\textbf{Opt~\uppercase\expandafter{\romannumeral3}: Presentation Cost Pruning.} As the total approximation cost of a rule will never be negative, given two rules $rule_1$ and $rule_2$, if the total presentation cost of $rule_2$ is greater than or equal to the total cost of $rule_1$, then the total cost of $rule_2$ must be greater than or equal to that of $rule_1$. We can leverage this as a pruning strategy to avoid the more expensive computations of the locally best rule.

\textbf{Globally Best Rule Algorithm.} Algorithm~\ref{alg:globalAppRule} presents the pseudo-code of the globally best rule calculation. We first calculate the values of $nz$ and $TotalCount$ based on $cd$ (Lines~\ref{alg:globalAppRule:init:start}-\ref{alg:globalAppRule:init:end}). $rule$ and $C_{Total}$ are the current globally best rule and its corresponding total cost, respectively. In Line~\ref{alg:globalAppRule:opt2}, by applying Opt~\uppercase\expandafter{\romannumeral2}, the number of bits $ln$ used for leading zeros encoding is enumerated from 0 to $ min(\lceil log_2{(nz + 1)}\rceil, 5)$ instead of 0 to 6. The Opt~\uppercase\expandafter{\romannumeral3} is employed in Lines~\ref{alg:globalAppRule:opt3:start}-\ref{alg:globalAppRule:opt3:end}. Specifically, once the total presentation cost $C_{Pre}$ of this loop is greater than or equal to the current minimum total cost $C_{Total}$, the loop is terminated. Otherwise, we consider the exponential value $z=2^{ln}$ in Line~\ref{alg:globalAppRule:opt1} (i.e., Opt~\uppercase\expandafter{\romannumeral1}). In Lines~\ref{alg:globalAppRule:call:start}-\ref{alg:globalAppRule:call:end}, we call Algorithm~\ref{alg:localAppRule} and update $rule$ and $C_{Total}$ if necessary. Finally, the globally best rule is returned.

\begin{algorithm}[t]
\small
	\caption{$GlobalAppRule(cd = \langle c_0, c_2, ..., c_{63}\rangle)$}\label{alg:globalAppRule}
	$nz \leftarrow |\{c_i \in \langle c_1, ..., c_{63}\rangle $ $|$ $ c_i \neq 0\}|$\;\label{alg:globalAppRule:init:start}
	$TotalCount \leftarrow \sum_{i = 0}^{63} c_i$\;\label{alg:globalAppRule:init:end}
	$rule \leftarrow \langle \ \rangle$; $C_{Total} \leftarrow \infty$\;
	\For(\tcp*[h]{Opt~\uppercase\expandafter{\romannumeral2}}){$ln \leftarrow 0 \text{\rm\ to } min(\lceil log_2{(nz + 1)}\rceil, 5)$\label{alg:globalAppRule:opt2}}{
		$C_{Pre} \leftarrow TotalCount \times ln$\;
		\If{$C_{Pre} \geq C_{Total}$\label{alg:globalAppRule:opt3:start}} {
			\Break;\label{alg:globalAppRule:opt3:end}	\tcp{Opt~\uppercase\expandafter{\romannumeral3}}	
		}
		$z \leftarrow 2^{ln}$; \label{alg:globalAppRule:opt1}\tcp{Opt~\uppercase\expandafter{\romannumeral1}}
		$(rule_{Local}, C_{App}) \leftarrow LocalAppRule(ca, z)$\;\label{alg:globalAppRule:call:start}
		\If{$C_{App} + C_{Pre} < C_{Total}$}{
			$C_{Total} \leftarrow C_{App} + C_{Pre}$\;
			$rule \leftarrow rule_{Local}$\;\label{alg:globalAppRule:call:end}
		}
	}
	\Return{$rule$}\;
\end{algorithm}

\textbf{Complexity Analysis.} Algorithm~\ref{alg:globalAppRule} calls Algorithm~\ref{alg:localAppRule} sequentially, so the \textbf{space complexity} of Algorithm~\ref{alg:globalAppRule} is the same with that of Algorithm~\ref{alg:localAppRule}, i.e., $\mathcal{O}(64 \times nz)$. We call Algorithm~\ref{alg:localAppRule} for at most $min(\lceil log_2{(nz + 1)}\rceil, 5) + 1 \leq 6$ times. As a result, the \textbf{time complexity} of Algorithm~\ref{alg:globalAppRule} is $\mathcal{O}(64 \times nz^3 \times 6)$. As shown in Fig.~\ref{fig:LeadingZerosCountDistribution}, for almost all datasets, $nz$ is smaller than 30. To this end, the computational cost of Algorithm~\ref{alg:globalAppRule} is acceptable, let alone there are a wide range of pruning strategies proposed by this paper.

\section{Optimization for Center Bits}\label{sec:centerbits}

In this section, we first point out the shortcomings of {\em Elf} for center bits encoding, and then propose to encode trailing zeros instead of center bits.

\subsection{Existing Solution for Encoding Center Bits}


To encode the center bits of an XORed value, all of Gorilla~\cite{pelkonen2015gorilla}, Chimp~\cite{liakos2022chimp} and {\em Elf}~\cite{li2023elf} need to write the number of center bits $center_t$, followed by the actual center bits. In particular, as shown in Fig.~\ref{fig:ElfEncodingStrategy}, if $center_t \leq 16$, {\em Elf} uses 4 bits for $center_t$; otherwise, {\em Elf} utilizes 6 bits. Although this strategy works for most time series data, there are still two main shortcomings. First, for the time series data with big decimal significand count, {\em Elf} can only erase a small number of last few bits, resulting in a large number of center bits, so it always requires 6 bits to encode $center_t$. Second, as it requires for an extra bit to distinguish between the two cases (i.e., ``$center_t \leq 16$'' and ``$center_t > 16$''), it in fact occupies as many as 5 or 7 bits to encode $center_t$.

\subsection{Encoding Trailing Zeros Instead of Center Bits}

Observing that the number of center bits $center_t$ of $xor'_t$ can be calculated by $center_t = 64 - lead_t - trail_t$, where $lead_t$ and $trail_t$ are the numbers of leading zeros and trailing zeros of $xor'_t$, respectively, this paper proposes to encode $trail_t$ (note that $lead_t$ is already stored) instead of $center_t$. Like $lead_t$, the distribution of $trail_t$ is extremely uneven, as shown in Fig.~\ref{fig:TrailingZerosCountDistribution}. To this end, we adopt the same technique for $lead_t$ introduced in Section~\ref{sec:leading} to encode $trail_t$.

\begin{figure}[t]
  \centering
  \includegraphics[width=3.5in]{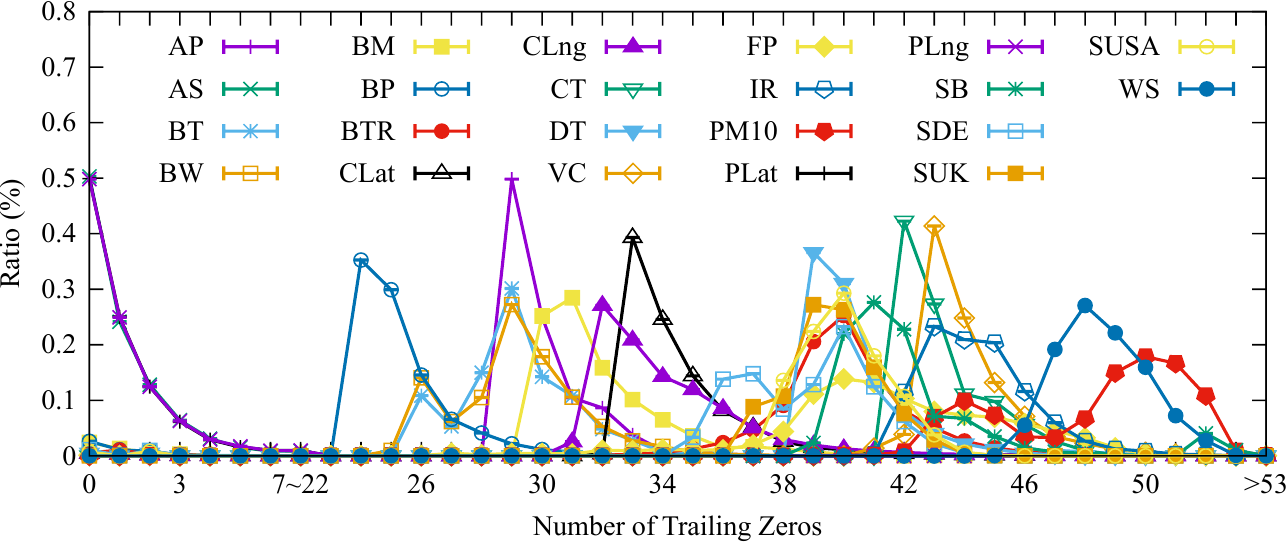}%
\vspace{-5pt}
  \caption{Distribution of Trailing Zeros Number.}
  \label{fig:TrailingZerosCountDistribution}
  \vspace{-10pt}
\end{figure}

Figure~\ref{fig:ElfStarEncoding}(a) exhibits the overall encoding strategies after performing the encoding optimizations for both leading zeros and center bits. Compared with Fig.~\ref{fig:ElfEncodingStrategy} with four cases, it has only three cases (i.e., we combine the two cases of ``$center_t \leq 16$'' and ``$center_t > 16$'' in Fig.~\ref{fig:ElfEncodingStrategy} into one). In particular, in the third case of Fig.~\ref{fig:ElfStarEncoding}(a) (i.e., ``$xor'_t \neq 0$ and not $C$''), we write $ln$ bits of $lead_t$ and $tn$ bits of $trail_t$, followed by the center bits with the length of $64 - lead_t - trail_t$~\footnote{In fact, we write $64 - App(lead_t) - App(trail_t)$ bits of center bits.}. Suppose the globally best rules for leading zeros and trailing zeros are $rule_{lead}$ and $rule_{trail}$, respectively, we let $ln = \lceil log_2{|rule_{lead}|}\rceil$ and $tn = \lceil log_2{|rule_{trail}|}\rceil$.

\begin{figure}[htb]
  \centering%
\vspace{-5pt}
  \includegraphics[width=3.5in]{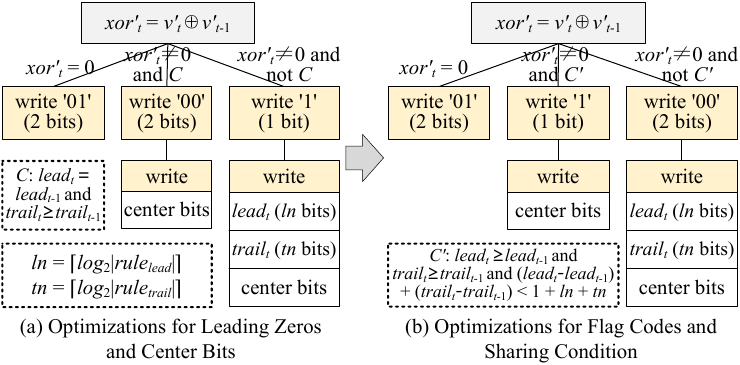}%
\vspace{-5pt}
  \caption{Evolution of Encoding Strategies.}%
  \label{fig:ElfStarEncoding}
  \vspace{-10pt}
\end{figure}

\section{Optimization for Sharing Condition}\label{sec:sharing}

In this section, we first point out the shortcomings if we simply apply the optimization for center bits to {\em Elf} encoding strategies, and then propose an adaptive sharing condition.

\subsection{Basic Solution}

Following the idea of {\em Elf}, if $xor'_t \neq 0$, there are two cases, as shown in Fig.~\ref{fig:ElfStarEncoding}(a). If the condition $C$ (i.e., ``$lead_t = lead_{t - 1}$ and $trail_t \geq trail_{t-1}$'') does not hold, we first write one flag bit of `1', then write $ln$ bits of $lead_t$ and $tn$ bits of $trail_t$, and finally write the center bits. If the condition $C$ holds, we first write two flag bits of `00', and then write the center bits directly. The rationale behind is to share the information of leading zeros and trailing zeros with the previous XORed value $xor'_{t-1}$, which is expected to save some bits.

However, it could be problematic for the following two reasons. First, it utilizes three flag codes with different lengths to distinguish the three cases. The selection of flag codes might be suboptimal, because it ignores the frequency of each case. Second, since we utilize the adaptive bits for $lead_t$ and $trail_t$, the condition $C$ adopted by {\em Elf} may be not suitable for our scenarios. The saved bits when sharing the information with the previous value may not offset the approximation cost. 

\subsection{Flag Codes Reassignment \& Sharing Condition Adaptation}

To ameliorate the problems above, this paper makes two corresponding optimizations, as shown in Fig.~\ref{fig:ElfStarEncoding}(b). 

First, we observe that for most time series, the erased XORed values that satisfy the condition $C$ are much more than those that do not, so it is not quite effective if we assign as many as two flag bits to the case of ``$xor'_t \neq 0$ and $C$''. Therefore, we switch the flag codes between case ``$xor'_t \neq 0$ and $C$'' and case ``$xor'_t \neq 0$ and not $C$''. 

Second, we devise an adaptive sharing condition technique. Specifically, we modify the condition $C$ to $C'$ (i.e., ``$lead_t \geq lead_{t-1}$ and $trail_t \geq trail_{t_1}$ and $(lead_t - lead_{t-1}) + (trail_t - trail_{t - 1}) < 1 + ln + tn$''). Here, the conditions of ``$lead_t \geq lead_{t-1}$ and $trail_t \geq trail_{t-1}$'' guarantee the lossless compression, and the condition ``$(lead_t - lead_{t-1}) + (trail_t - trail_{t - 1}) < 1 + ln + tn$'' ensures a positive gain after we share the previous information. ``$1 + ln + tn$'' is the extra bits (including the flag codes) we need to write when $C'$ dose not hold compared to the case when $C'$ holds.

\section{Streaming Scenarios Extension}\label{sec:stream}

In order to obtain the rules for leading zeros and trailing zeros, it requires to know in advance the distributions of leading zeros and trailing zeros, which is not practical in streaming scenarios. In this section, we extend {\em Elf}$^*$ to Streaming {\em Elf}$^*$, i.e., {\em SElf}$^*$.

\setlength{\tabcolsep}{0.33em} 
\begin{table}[b]
\vspace{-15pt}
\caption{Distribution Euclidean Distance of Different Time Windows}
\centering\label{tbl:streamObservation}
\begin{tabular}{|c|c||c|c|} 
\hline
\multicolumn{2}{|c||}{\textbf{Leading Zeros}}	&\multicolumn{2}{c|}{\textbf{Trailing Zeros}}\\
\hline
Same Dataset	& Different Datasets	& Same Dataset	& Different Datasets\\
\hline
0.15			& 0.46					& 0.22			& 0.66				\\
\hline			
\end{tabular}
\end{table}

\textbf{Observation.} To investigate the distribution changes of leading zeros and trailing zeros, we split each dataset into multiple time windows, where each window contains 1,000 consecutive values. In each window, we calculate a leading zeros distribution and a trailing zeros distribution. We compute the Euclidean distances of the corresponding distributions between windows pairwise. Table~\ref{tbl:streamObservation} reports the average Euclidean distances between windows in the same dataset or different datasets, respectively, from which we can conclude that, for both leading zeros and trailing zeros, their distributions are much more similar in the same dataset than in different datasets. It motivates us to calculate the rules using the previous values in the same dataset in streaming scenarios.

\textbf{Streaming {\em Elf}$^*$.} In streaming scenarios, we employ a sliding window technique. As shown in Fig.~\ref{fig:SElfStar}(a), during compression, in the $i$-th window, we first write the rules $rule_{lead}^i$ and $rule_{trail}^i$, and then streamingly encode up to $w$ erased XORed values, where $w$ is the window size. During decompression, we first read $rule_{lead}^i$ and $rule_{trail}^i$, with which we can decompress the values in this window successfully.

\begin{figure}[t]
  \centering
  \includegraphics[width=3.5in]{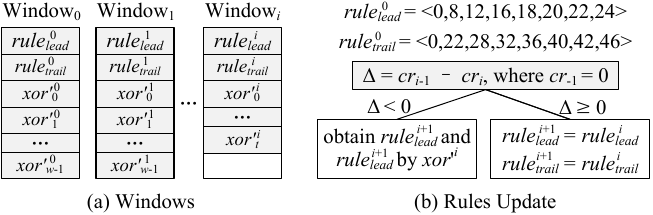}%
\vspace{-5pt}
  \caption{Streaming {\em Elf}$^*$.}
  \label{fig:SElfStar}
\vspace{-10pt}
\end{figure}

The calculation of the rules in a window without knowing the distributions in this window is presented in Fig.~\ref{fig:SElfStar}(b). We set fixed $rule_{lead}^0$ and $rule_{trail}^0$ manually. However, they may not be effective for different datasets or for all windows in the same time series. To compute $rule_{lead}^{i+1}$ and $rule_{trail}^{i+1}$ in the upcoming $(i$+1$)$-th window, this paper devises a simple but effective solution. Particularly, we compute the compression ratios (i.e., $cr_{i-1}$ and $cr_{i}$) of the $(i$-1$)$-th and $i$-th windows, respectively, and then compute their difference $\Delta = cr_{i-1} - cr_i$. If the performance deteriorates in terms of compression ratio (i.e., $\Delta < 0$), we compute $rule_{lead}^{i+1}$ and $rule_{trail}^{i+1}$ using the values in the $i$-th window; otherwise, we share the rules with the $i$-th window directly. This solution not only ensures high compression ratio, but also reduces the computational cost. The effect of window size selection on compression performance is studied in Section~\ref{sec:exp}.

\begin{table*}[t]
\caption{Details of Datasets}
\centering\label{tbl:datasets}
\resizebox{6.8in}{!}{
\begin{tabular}{|c|c|c|c|c|l|} 
\hline
\multicolumn{2}{|c|}{\textbf{Dataset}}	&\textbf{\#Records}	&\textbf{$\beta$}	&\textbf{Time Span}  &\textbf{Description}\\
\hline
\hline
\multirow{14}{*}{\rotatebox[origin=c]{90}{\textbf{Time Series}}}		
&Air-pressure (AP)~\cite{AirPressure}	    &137,721,453		&7					&6 years 	& Barometric pressure corrected by  sea level and surface level\\
&Air-sensor (AS)~\cite{influxdb2data}		&8,664				&17					&1 hour 	& A synthetic dataset that includes air sensor data with random noise\\
&Bird-migration (BM)~\cite{influxdb2data}	&17,964				&7					&1 year 	& The changes in the location of bird migration and vegetation cover\\
&Bitcoin-price (BP)~\cite{influxdb2data}	&2,741				&9					&1 month 	& Bitcoin's price fluctuation in dollar exchange rate\\
&Basel-temp (BT)~\cite{Basel}				&124,079			&9					&14 years  	& The air temperature in Basel, Switzerland\\
&Basel-wind (BW)~\cite{Basel}				&124,079			&8					&14 years 	& The wind speed in Basel, Switzerland\\
&City-temp (CT)~\cite{CityTemp}				&2,905,887			&3					&25 years  	& The temperature in numerous cities all over the world\\
&Dewpoint-temp (DT)~\cite{DewpointTemp}		&5,413,914			&4					&3 years 	& The relative dew point measured by sensors floating on rivers and lakes\\
&IR-bio-temp (IR)~\cite{IRBioTemp}			&380,817,839		&3					&7 years 	& Infrared biological temperature\\
&PM10-dust (PM10)~\cite{PM10Dust}			&222,911			&3					&5 years 	& Concentration of PM10 in the atmosphere\\
&Stocks-DE (SDE)~\cite{Stocks}				&45,403,710			&6					&1 year 	& The price fluctuation of stock exchanges in Germany\\
&Stocks-UK (SUK)~\cite{Stocks}				&115,146,731		&5					&1 year	 	& The price fluctuation of stock exchanges in UK\\
&Stocks-USA (SUSA)~\cite{Stocks}			&374,428,996		&4					&1 year 	& The price fluctuation of stock exchanges in USA\\
&Wind-speed (WS)~\cite{WindDir}				&199,570,396		&2					&6 years 	& The speed of horizontal and vertical winds\\

\hline
\hline
\multirow{8}{*}{\rotatebox[origin=c]{90}{\textbf{Non-Time Series}}}
&Blockchain-tr (BTR)~\cite{BlockchianTr}	&231,031			&5					&- 			& Bitcoin daily transaction volume\\
&City-lat (CLat)~\cite{citylat}				&41,001				&6					&- 			& The latitudes of the cities and towns  around the world\\
&City-lon (CLon)~\cite{citylat}				&41,001				&7					&- 			& The longitudes of the cities and towns around the world\\
&Food-price (FP)~\cite{WorldFoodPrice}		&2,050,638			&3					&- 			& The changes of global food prices\\
&POI-lat (PLat)~\cite{POI}					&424,205			&16					&- 			& The latitudes in radian coordinates of POIs extracted from Wikipedia\\
&POI-lon (PLon)~\cite{POI}					&424,205			&16					&- 			& The longitudes in radian coordinates of POIs extracted from Wikipedia\\
&SD-bench (SB)~\cite{SSD}					&8,927				&4					&- 			& The performance of multiple storage drives scored by benchmarks\\
&Vehicle-charge (VC)~\cite{VehicleCharge}	&3,395				&3					&- 			& The total energy usage and charging time for a group of electric vehicles\\

\hline
\end{tabular}
}
\vspace{-10pt}
\end{table*}

\section{Experiments}\label{sec:exp}


\subsection{Experimental Settings}
\textbf{Datasets.} We use 22 datasets to evaluate the performance of {\em Elf}$^*$ and {\em SElf}$^*$. These datasets consist of 14 time series and 8 non-time series, all of which are consistent with those in {\em Elf} paper~\cite{li2023elf}. The details of the datasets are listed in Table~\ref{tbl:datasets}. In this paper, we regard $\beta \geq 16$ as large, and $\beta < 16$ as small.

\textbf{Baselines.} We evaluate the performance of {\em SElf}$^*$ in comparison to \textbf{five} streaming floating-point compression methods, i.e., Gorilla~\citep{pelkonen2015gorilla}, Chimp~\cite{liakos2022chimp}, Chimp$_{128}$~\citep{liakos2022chimp}, {\em Elf}~\cite{li2023elf} and {\em Elf+}~\cite{Elf+}. Here, {\em Elf}+ is an upgraded version of {\em Elf}, which optimizes the encoding strategies for the significand counts based on the fact that the values in a time series are expected to have the same significand count. All these optimizations proposed by {\em Elf}+ are also adopted by this paper.
Meanwhile, we compare {\em Elf}$^*$ with \textbf{three} widely-used batch compression methods, including Xz~\cite{Xz}, Zstd~\cite{collet2016zstd} and Snappy~\cite{snappy}. Besides, we conduct ablation experiments with \textbf{three} variants, i.e., {\em Elf$^*$-S} ({\em Elf}$^*$ without optimizations for sharing condition), {\em Elf$^*$-S-C} ({\em Elf}$^*$ without optimizations for sharing condition and center bits), and {\em Elf}+ (note that without any optimizations, {\em Elf}$^*$ is degraded into {\em Elf}+). All experiments are implemented in Java, and the source codes are publicly released~\cite{ElfStar}.

\textbf{Metrics.} We use three metrics, i.e., compression ratio, compression time and decompression time, to measure the performance. It is worth noting that the compression ratio here refers to the ratio of the compressed data size to the uncompressed data size.

\textbf{Settings.} As these works~\cite{liakos2022chimp, li2023elf, Elf+} did, by default, we regard 1,000 consecutive records in each dataset as a block. Each compression method is executed on up to 100 different blocks for each dataset, and the average measurements of one block are finally reported. The default window size for {\em SElf}$^*$ is 1,000. We run our experiments on a computer with Windows 11, 11th Gen Intel (R) Core (TM) i5-11400 @ 2.60GHz CPU and 16GB memory. The JDK (Java Development Kit) version is 1.8.

\subsection{Overall Comparison with Baselines}

\setlength{\tabcolsep}{0.1em} 
\begin{table*}[t]
\caption{Overall Comparison with Baselines (the best results are in \textbf{bold}, while the suboptimal results are in \underline{underlined}).} 
\centering\label{tbl:overallcomparison}
\resizebox{7.16in}{!}{
\begin{tabular}{|c|c|c||c|c|c|c|c|c|c|c|c|c|c|c|c|c||c||c|c|c|c|c|c|c|c||c|} 
\hline
\multicolumn{3}{|c||}{\multirow{2}{*}{\textbf{Dataset}}} & \multicolumn{15}{c||}{\textbf{Time Series}} & \multicolumn{9}{c|}{\textbf{Non-Time Series}} \\
\cline{4-27}


\multicolumn{3}{|c||}{}& AP	& AS & BM & BP & BT	& BW	& CT	 & DT & IR & PM10 & SDE & SUK & SUSA	 & WS & Avg. & BTR & CLat & CLon	&  FP & PLat & PLon &SB & VC	 &  Avg.\\
\hline
\hline

\multirow{10}{*}{\rotatebox[origin=c]{90}{\textbf{Compression Ratio}}}	&  	\multirow{6}{*}{\rotatebox[origin=c]{90}{\textbf{Streaming}}}& Gorilla&0.719 &	0.822 &	0.786 &	0.835 &	0.940 &	0.994 &	0.847 &	0.833 &	0.696 &	0.478 &	0.717 &	0.577 &	0.678 &	0.828 &	0.768 &	0.737 &	1.032 &	1.032 &	0.576 &	1.030 &	1.032 &	0.629 &	0.996 &	0.883 \\
& & Chimp&0.653 &	\textbf{0.774} &	0.716 &	0.766 &	0.846 &	0.876 &	0.641 &	0.775 &	0.640 &	0.427 &	0.670 &	0.519 &	0.640 &	0.817 &	0.697 &	0.668 &	0.923 &	0.984 &	0.471 &	\underline{0.901} &	\underline{0.989} &	0.548 &	0.862 &	0.793 \\
& & Chimp$_{128}$&0.544 &	\textbf{0.774} &	0.496 &	0.722 &	\textbf{0.472} &	0.712 &	0.318 &	0.347 &	0.248 &	0.228 &	0.271 &	0.286 &	0.230 &	0.232 &	0.420 &	0.552 &	0.777 &	0.849 &	0.340 &	\textbf{0.898} &	\textbf{0.986} &	0.266 &	0.362 &	0.629 \\
& & Elf&0.306 &	0.848 &	0.423 &	0.565 &	0.579 &	0.587 &	0.255 &	0.311 &	0.220 &	0.173 &	0.262 &	0.219 &	0.240 &	0.256 &	0.375 &	0.359 &	0.559 &	0.628 &	0.232 &	0.962 &	1.059 &	0.267 &	0.342 &	0.551 \\
& & Elf+&\underline{0.255} &	0.863 &	\underline{0.378} &	\underline{0.497} &	0.524 &	\underline{0.556} &	\underline{0.217} &	\underline{0.256} &	\underline{0.163} &	\underline{0.122} &	\underline{0.229} &	\underline{0.193} &	\underline{0.184} &	\underline{0.204} &	\underline{0.331} &	\underline{0.300} &	\underline{0.506} &	\underline{0.600} &	\underline{0.220} &	0.977 &	1.074 &	\underline{0.235} &	\underline{0.290} &	\underline{0.525} \\
& & SElf*& \textbf{0.234} &	0.797 &	\textbf{0.343} &	\textbf{0.443} &	\textbf{0.472} &	\textbf{0.504} &	\textbf{0.170} &	\textbf{0.213} &	\textbf{0.162} &	\textbf{0.121} &	\textbf{0.197} &	\textbf{0.170} &	\textbf{0.159} &	\textbf{0.169} &	\textbf{0.297} &	\textbf{0.272} &	\textbf{0.453} &	\textbf{0.546} &	\textbf{0.204} &	0.915 &	1.001 &	\textbf{0.221} &	\textbf{0.260} &	\textbf{0.484} \\
\cline{2-27}

& \multirow{4}{*}{\rotatebox[origin=c]{90}{\textbf{Batch}}} & Elf*&\textbf{0.232} &	\textbf{0.786} &	\textbf{0.336} &	\textbf{0.431} &	0.466 &	\textbf{0.502} &	\textbf{0.169} &	\textbf{0.208} &	\textbf{0.130} &	\textbf{0.106} &	\textbf{0.191} &	\textbf{0.153} &	\textbf{0.144} &	\underline{0.163} &	\textbf{0.287} &	\textbf{0.267} &	\textbf{0.451} &	\textbf{0.544} &	\textbf{0.199} &	\textbf{0.914} &	1.000 &	0.216 &	\underline{0.247} &	\textbf{0.480} \\
& & Xz&\underline{0.461} &	\textbf{0.786} &	\underline{0.428} &	\underline{0.627} &	\textbf{0.347} &	\underline{0.573} &	\underline{0.179} &	\underline{0.272} &	\underline{0.167} &	\underline{0.114} &	\underline{0.192} &	\underline{0.161} &	\underline{0.171} &	\textbf{0.149} &	\underline{0.331} &	\underline{0.400} &	\underline{0.602} &	\underline{0.632} &	\underline{0.231} &	\underline{0.928} &	\textbf{0.959} &	\textbf{0.127} &	\textbf{0.231} &	\underline{0.514} \\
& & Zstd&0.575 &	0.914 &	0.506 &	0.753 &	\underline{0.407} &	0.607 &	0.221 &	0.380 &	0.254 &	0.151 &	0.265 &	0.221 &	0.240 &	0.192 &	0.406 &	0.452 &	0.676 &	0.708 &	0.301 &	0.940 &	\underline{0.961} &	\underline{0.166} &	0.339 &	0568 \\
& & Snappy&0.720 &	1.002 &	0.610 &	0.987 &	0.539 &	0.753 &	0.291 &	0.505 &	0.315 &	0.216 &	0.345 &	0.316 &	0.323 &	0.275 &	0.514 &	0.541 &	0.834 &	0.873 &	0.393 &	1.002 &	1.002 &	0.245 &	0.420 &	0.664 \\
\hline
\hline

\multirow{10}{*}{\rotatebox[origin=c]{90}{\textbf{Compression Time ($\mu$s)}}}	&  	\multirow{6}{*}{\rotatebox[origin=c]{90}{\textbf{Streaming}}}& Gorilla&\textbf{21} &	\textbf{40} &	\textbf{20} &	\textbf{21} &	\textbf{31} &	\textbf{24} &	\textbf{21} &	\textbf{20} &	\textbf{22} &	\textbf{18} &	\textbf{19} &	\textbf{19} &	\textbf{19} &	\textbf{20} &	\textbf{22} &	\textbf{20} &	\textbf{22} &	\textbf{25} &	\textbf{20} &	\textbf{23} &	\textbf{25} &	\textbf{20} &	\textbf{25} &	\textbf{23} \\
& & Chimp&\textbf{21} &	\underline{41} &	\underline{27} &	\underline{30} & \textbf{31} &	\underline{29} &	\underline{28} &	\underline{27} &	\underline{24} &	\underline{21} &	\underline{25} &	\underline{27} &	\underline{25} &	\underline{26} &	\underline{27} &	\underline{25} &	\underline{28} &	\underline{30} &	\underline{25} &	\underline{28} &	\underline{34} &	\underline{26} &	\underline{34} &	\underline{29} \\
& & Chimp$_{128}$&31 &	52 &	33 &	40 &	38 &	41 &	30 &	32 &	26 &	24 &	27 &	30 &	26 &	\underline{26} &	33 &	32 &	40 &	45 &	31 &	35 &	40 &	31 &	38 &	37 \\
& & Elf&48 &	271 &	65 &	63 &	120 &	90 &	50 &	55 &	51 &	47 &	58 &	54 &	57 &	62 &	78 &	90 &	64 &	74 &	53 &	79 &	99 &	55 &	62 &	72 \\
& & Elf+&70 &	209 &	51 &	48 &	114 &	102 &	36 &	42 &	35 &	30 &	43 &	46 &	39 &	40 &	65 &	93 &	56 &	68 &	42 &	91 &	75 &	37 &	48 &	64 \\
& & SElf*&67 &	206 &	67 &	65 &	101 &	90 &	41 &	45 &	46 &	40 &	51 &	50 &	46 &	50 &	69 &	74 &	58 &	76 &	61 &	68 &	75 &	50 &	51 &	64 \\
\cline{2-27}

& \multirow{4}{*}{\rotatebox[origin=c]{90}{\textbf{Batch}}} & Elf*&\textbf{81} &	\underline{162} &	\textbf{93} &	\textbf{82} &	\textbf{97} &	\textbf{90} &	\textbf{60} &	\textbf{63} &	\textbf{60} &	\textbf{56} &	\textbf{69} &	\textbf{69} &	\textbf{61} &	\textbf{63} &	\textbf{79} &	\textbf{81} &	\textbf{76} &	\textbf{93} &	\textbf{89} &	\textbf{69} &	\underline{84} &\textbf{	73} &	\textbf{72} &	\textbf{80} \\
& & Xz&1974 &	2411 &	1820 &	1923 &	1565 &	2036 &	1305 &	1542 &	1294 &	1107 &	1277 &	1190 &	1238 &	1247 &	1566 &	1562 &	1685 &	1795 &	1446 &	1938 &	1999 &	1250 &	1779 &	1682 \\
& & Zstd&270 &	211 &	271 &	281 &	262 &	274 &	213 &	249 &	169 &	208 &	188 &	160 &	161 &	108 &	216 &	257 &	250 &	232 &	200 &	215 &	194 &	156 &	199 &	213 \\
& & Snappy&\underline{123} &	\textbf{114} &	\underline{114} &	\underline{113} &	\underline{113} &	\underline{125} &	\underline{93} &	\underline{102} &	\underline{87} &	\underline{93} &	\underline{104} &	\underline{82} &	\underline{82} &	\underline{82} &	\underline{102} &	\underline{96} &	\underline{98} &	\underline{94} &	\textbf{89} &	\underline{76} &	\textbf{80} &	\underline{83} &	\underline{114} &	\underline{91} \\
\hline
\hline

\multirow{10}{*}{\rotatebox[origin=c]{90}{\textbf{Decompression Time ($\mu$s)}}}	&  	\multirow{6}{*}{\rotatebox[origin=c]{90}{\textbf{Streaming}}} & Gorilla&\textbf{20} &	\textbf{28} &	\textbf{20} &	\textbf{21} &	\underline{39} &	\textbf{22} &	\underline{20} &	\textbf{20} &	\underline{21} &	\underline{18} &	\underline{20} &	\textbf{18} &	\underline{20} &	\underline{20} &	\textbf{22} &	\textbf{19} &	\textbf{20} &	\textbf{23} &	\textbf{19} &	\textbf{21} &	\textbf{20} &	\textbf{20} &	\underline{25} &\textbf{21} \\
& & Chimp&\underline{22} &	\underline{34} &	26 &	\underline{27} &	43 &	\underline{25} &	24 &	25 &	23 &	19 &	24 &	22 &	23 &	24 &	26 &	\underline{23} &	\underline{25} &	\underline{27} &	22 &	\underline{23} &	\underline{26} &	24 &	32 &	\underline{25} \\
& & Chimp$_{128}$&30 &	39 &	\underline{24} &	29 &	\textbf{34} &	29 &	\textbf{19} &	\textbf{20} &	\textbf{18} &	\textbf{17} &	\textbf{19} &	\underline{20} &	\textbf{18} &	\textbf{18} &	\underline{24} &	25 &	29 &	32 &	\underline{20} &	26 &	28 &	\textbf{20} &	\textbf{24} &	\underline{25} \\
& & Elf&49 &	98 &	47 &	51 &	88 &	55 &	41 &	45 &	43 &	41 &	43 &	35 &	44 &	47 &	52 &	63 &	50 &	57 &	37 &	29 &	32 &	39 &	47 &	44 \\
& & Elf+&34 &	120 &	38 &	37 &	84 &	45 &	36 &	34 &	32 &	29 &	33 &	30 &	32 &	34 &	44 &	64 &	45 &	55 &	33 &	38 &	37 &	33 &	44 &	44 \\
& & SElf*&27 &	46 &	38 &	37 &	49 &	46 &	27 &	27 &	30 &	27 &	33 &	28 &	29 &	31 &	34 &	34 &	39 &	52 &	31 &	33 &	34 &	30 &	37 &	36 \\
\cline{2-27}

& \multirow{3}{*}{\rotatebox[origin=c]{90}{\textbf{Batch}}}& Elf*&\textbf{28} &	\underline{54} &	\textbf{40} &	\textbf{37} &	\textbf{43} &	\underline{45} &	\textbf{28} &	\textbf{28} &	\underline{29} &	\underline{27} &	33 &	\underline{28} &	\underline{29} &	31 &	\textbf{34} &	\textbf{35} &	\textbf{39} &	52 &	\underline{31} &	\underline{33} &	34 &	\underline{31} &	\textbf{35} &	\underline{36} \\
& & Xz&398 &	665 &	349 &	511 &	291 &	456 &	156 &	215 &	135 &	97 &	152 &	139 &	137 &	120 &	273 &	315 &	453 &	469 &	193 &	666 &	687 &	111 &	195 &	386 \\
& & Zstd&59 &	58 &	60 &	60 &	52 &	57 &	48 &	\underline{38} &	30 &	29 &	\underline{30} &	31 &	30 &	\underline{29} &	44 &	59 &	49 &	\underline{50} &	39 &	\underline{33} &	\underline{29} &	\underline{31} &	48 &	42 \\
& & Snappy&\underline{47} &	\textbf{38} &	\underline{48} &	\underline{41} &	\underline{47} &	\textbf{42} &	\underline{38} &	41 &	\textbf{25} &	\textbf{24} &	\textbf{25} &	\textbf{25} &	\textbf{25} &	\textbf{23} &	\underline{35} &	\underline{43} &	\underline{43} &	\textbf{44} &	\textbf{30} &	\textbf{19} &	\textbf{19} &	\textbf{23} &	\underline{39} &	\textbf{33} \\
\hline

\end{tabular}
}
\vspace{-10pt}
\end{table*}

Table~\ref{tbl:overallcomparison} presents the results of various compression algorithms on different datasets. We group the algorithms into streaming and batch ones. Moreover, we divide the datasets into time series and non-time series. We investigate the performance of different algorithms in each group separately.

\subsubsection{Compression Ratio} According to Table~\ref{tbl:overallcomparison}, we have the following observations in terms of compression ratio.

\textbf{{\em SElf}$^*$ VS Other Streaming Algorithms}. Among the streaming compression algorithms, {\em SElf}$^*$ usually performs the best in terms of compression ratio. For example, for time series datasets, compared with the best streaming competitor {\em Elf}+, {\em SElf}$^*$ achieves an average improvement of 10.27\% (i.e., $(0.331-0.297)/0.331 \approx 10.27\%$) with regard to compression ratio. Similarly, for non-time series datasets, {\em SElf}$^*$ has an average compression ratio improvement of 7.8\% over {\em Elf}+. All of these verify the effectiveness of the optimizations proposed by this paper. We also notice that {\em SElf}$^*$ enjoys about 61.3\% and 29.3\% relative improvement over Gorilla and Chimp$_{128}$ for time series datasets, respectively, which proves the powerful compression capability of {\em SElf}$^*$. For the datasets with large significand count $\beta$ (i.e., AS, PLat and PLng), {\em SElf}$^*$ performs slightly worse than Chimp$_{128}$. This is because {\em SElf}$^*$ is based on {\em Elf}. {\em Elf} would erase fewer bits for the datasets with larger $\beta$, resulting in an unsatisfactory compression ratio. However, this gap is narrowed by {\em SElf}$^*$. {\em SElf}$^*$ has the same best compression ratio as Chimp$_{128}$ on BT, although the significand count of BT is only 9. We find that there are a large number of equal values in BT, so Chimp$_{128}$ can always find the same value from the previous 128 records.

\textbf{{\em Elf}$^*$ VS Other Batch Algorithms}. {\em Elf}$^*$ shows the best compression ratio among the batch compression algorithms for almost all datasets. Specifically, compared with the best competitor Xz, {\em Elf}$^*$ achieves relative average improvement of 13.3\% and 6.6\% for time series and non-time series, respectively. For the datasets of BT and SB, Xz outperforms {\em Elf}$^*$ significantly. It could be the reason that the records in BT and SB contain many same substrings (although their values are not exactly the same, e.g., 2.57\underline{05285} vs 2.61\underline{05285}). Thus, it is of great benefit for Xz to build a dictionary and reach a better compression ratio. {\em Elf}$^*$ outperforms Zstd and Snappy by average relative improvement of 29.3\% and 44.2\% for time series datasets, respectively. For non-time series datasets, the improvement over Zstd and Snappy turns out to be 15.5\% and 27.7\%, respectively. Compared with time series, the improvement of {\em Elf}$^*$ over other batch algorithms for non-time series is not so significant, which indicates the characteristics of time series can enhance {\em Elf}$^*$ drastically.

\textbf{{\em SElf}$^*$ VS {\em Elf}$^*$}. For all datasets, the compression ratio of {\em SElf}$^*$ is slightly worse than that of {\em Elf}$^*$, but the difference can be negligible. On one hand, the distributions of leading zeros and trailing zeros in the same dataset do not fluctuate significantly. On the other hand, we apply a rigorous criterion for rules update. Once the compression ratio shows a trend of deterioration, we update the rules based on the latest values, which guarantees a good compression performance.

\subsubsection{Compression Time and Decompression Time}
With regard to compression time and decompression time, we have the following findings.

\textbf{{\em SElf}$^*$ VS Other Streaming Algorithms}. Compared with Gorilla, Chimp and Chimp$_{128}$, the algorithms in {\em Elf}'s family (i.e., {\em Elf}, {\em Elf+}, {\em Elf}$^*$ and {\em SElf}$^*$) spend a bit more compression time and decompression time due to the erasing steps and restoring steps, which inevitably take some time. Among the algorithms in {\em Elf}'s family, {\em Elf+} has the shortest compression time, because it utilizes the significand count of the previous value and employs an efficient numerical checking method. In despite of the extra computational costs for approximation rules, {\em SElf}$^*$ manifests a similar compression time as {\em Elf}+. There could be two reasons. First, we propose a set of pruning strategies during the computation for rules and employs an effective rule update strategy, which reduces the compression time considerably. Second, {\em SElf}$^*$ has a better compression ratio, indicating fewer compressed bits to write. The decompression logic of {\em SElf}$^*$ is the same with that of {\em Elf}+, but {\em SElf}$^*$ reads fewer bits, resulting in a shorter decompression time.

\textbf{{\em Elf}$^*$ VS Other Batch Algorithms}. Among batch compression algorithms, {\em Elf}$^*$ has the shortest compression time for almost all datasets. For example, the average compression time of Xz is 20 times that of {\em Elf}$^*$ for both time series and non-time series.  Besides, {\em Elf}$^*$ takes only one-ninth of the decompression time compared to Xz. Compared with compression time, the difference of decompression time is not so significant. To this end, most compression algorithms focus more on the trade-off between compression ratio and compression time.

\textbf{{\em SElf}$^*$ VS {\em Elf}$^*$}. Compared with {\em Elf}$^*$, {\em SElf}$^*$ enjoys less compression time, because it can share the rules with the previous time window, thus avoiding computing them each time.  {\em SElf}$^*$ and {\em Elf}$^*$ have almost the same decompression time, since their decompression logics are exactly the same.
 
\subsubsection{Summary}
When considering the compression ratio, compression time and decompression time comprehensively, {\em SElf}$^*$ and {\em Elf}$^*$ can usually achieve remarkable performance.

\subsection{Ablation Experiments}

\begin{figure}[t]
	\centering
	\includegraphics[width=3.4in]{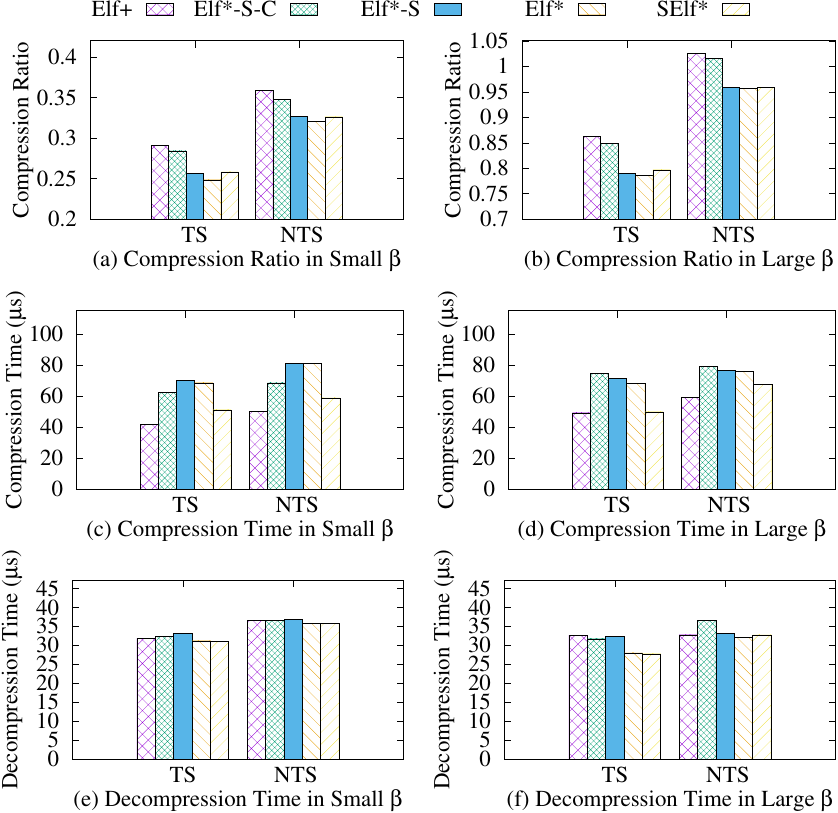}%
	\caption{Performance in Ablation Experiments} %
	\label{fig:exp:ablation}%
\vspace{-10pt}
\end{figure}%

To verify the effectiveness of the optimizations for leading zeros, center bits and sharing condition, we conduct a set of ablation experiments. 

\subsubsection{Compression Ratio}

It is observed from Figures~\ref{fig:exp:ablation}(a-b) that, for both time series (i.e., TS) and non-time series (i.e., NTS) with small or big $\beta$, {\em Elf$^*$-S-C} has a better compression ratio than {\em Elf+} , indicating the effectiveness of the optimizations for leading zeros. Similarly, {\em Elf$^*$-S} manifests a better compression ratio over {\em Elf$^*$-S-C}, which verifies the effectiveness of center bits optimizations. Thanks to the optimizations for sharing condition, {\em Elf}$^*$ further exhibits a better compression ratio compared to {\em Elf$^*$-S}. {\em SElf}$^*$ shows a slightly larger compression ratio than {\em Elf}$^*$, since it computes the approximation rules using the previous values, which may introduce some errors. It is interesting to observe that the compression ratio improvement of center bits optimizations (i.e., {\em Elf$^*$-S} over {\em Elf$^*$-S-C}) is the most significant compared to other optimizations. This is because {\em Elf$^*$-S-C} utilizes as many as 5 or 7 bits to encode the numbers of center bits, while {\em Elf$^*$-S} leverages much fewer bits (usually less than 3 bits) to encode the trailing zeros with adaptive approximation rules. 

Besides, for all algorithms in {\em Elf}'s family, the compression ratio on time series datasets is much better than that on non-time series datasets, since these algorithms successfully capture the characteristic that two consecutive values are expected to be similar, which are supposed to produce XORed values with many leading zeros. Moreover, datasets with large $\beta$ usually have larger compression ratio than those with small $\beta$, due to fewer bits can be erased. 

\subsubsection{Compression Time}
As shown in Figures~\ref{fig:exp:ablation}(c-d), {\em Elf$^*$-S-C} has a longer compression time than {\em Elf}+, because it needs to compute the leading zeros distributions and approximation rules. Although {\em Elf$^*$-S} takes more compression time than {\em Elf$^*$-S-C} when $\beta$ is small, the gap between {\em Elf$^*$-S} over {\em Elf$^*$-S-C} is narrower than that between {\em Elf$^*$-S-C} over {\em Elf}+. There could be two reasons. First, we share the calculation process of leading zeros distribution and trailing zeros distribution. Second, since {\em Elf$^*$-S} can achieve a much better compression ratio than {\em Elf$^*$-S-C}, it takes much less time for {\em Elf$^*$-S} to write fewer compressed bits. {\em Elf}$^*$ takes less compression time than {\em Elf$^*$-S}, because the sharing condition optimization does not bring in any extra computational cost but can improve the compression ratio, leading to fewer bits to write. {\em SElf}$^*$ takes much less time than {\em Elf}$^*$ because it avoids some computations for approximation rules.

\subsubsection{Decompression Time}
Figures~\ref{fig:exp:ablation}(e-f) show that all algorithms in {\em Elf}'s family exhibit similar decompression time. Nonetheless, {\em Elf}$^*$ and {\em SElf}$^*$ spend slightly less decompression time than other algorithms. Because both two algorithms have better compression ratio, resulting in fewer bits to read.

\subsection{Performance with Different Block Sizes}

We conduct a set of experiments to investigate the effect of different block sizes on batch compression algorithms (note that the concept of block is actually not applicable to streaming scenarios). As shown in Fig.~\ref{fig:exp:block}(a), with an increasing block size, the compression ratio of {\em Elf}$^*$ keeps stable, while that of Xz, Zstd and Snappy first drops sharply and then decreases slowly when the block size is greater than 1,000. As shown in Fig.~\ref{fig:exp:block}(b), the compression efficiency of all algorithms first decreases rapidly and then keeps stable. When a block contains more data, Xz, Zstd and Snappy can build better dictionaries with less average time and fewer average bits per record.

\begin{figure}[t]
	\centering
	\includegraphics[width=3.5in]{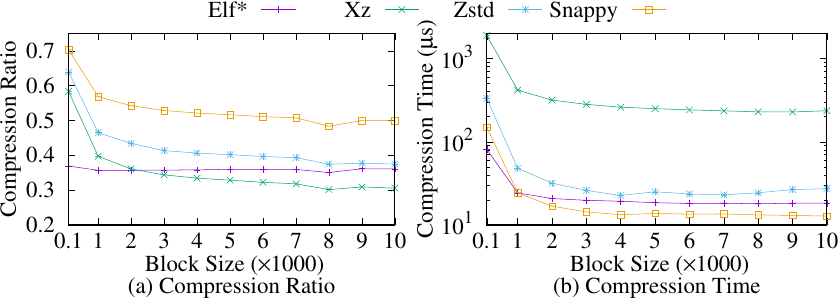}%
	\caption{Average Performance on All Datasets VS Different Block Sizes (the compression time is normalized to the time when compressing 1,000 records)} %
	\label{fig:exp:block}%
\end{figure}%

{\em Elf}$^*$ always performs better than Zstd in terms both of compression ratio and compression time. If the block size is larger than 2,000, the compression ratio of Xz is better than that of {\em Elf}$^*$, but it takes up to 2 orders of magnitude more compression time compared to {\em Elf}$^*$. For Snappy, although it is the fastest algorithm when the block size is larger than 1,000, its compression ratio is much worse than that of {\em Elf}$^*$.

Overall, when the block size is small, {\em Elf}$^*$ usually performs the best among all compared algorithms. When the block size is large, {\em Elf}$^*$ is still one of the most competitive algorithms.

\subsection{Performance with Different Window Sizes}

\begin{figure}[t]
\centering
\includegraphics[width=3.5in]{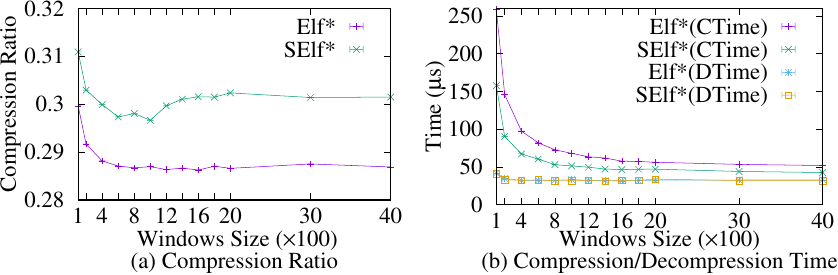}%
\caption{Average Performance on Time Series VS Different Window Sizes (the time is normalized to the time when handling 1,000 records).} %
\label{fig:exp:windowsize}%
\vspace{-10pt}
\end{figure}%

To explore the effect of window size selection on the performance of {\em SElf}$^*$, we conduct a set of experiments using the time series datasets (non-time series data are usually not generated in a streaming fashion). {\em Elf}$^*$ is selected as the baseline, since its compression performance provides a bound for {\em SElf}$^*$. {\em Elf}$^*$ regards the data in a time window as a block, and updates the rules every time. 

As shown in Fig.~\ref{fig:exp:windowsize}(a), when the window size gets larger, the compression ratio of {\em Elf}$^*$ first decreases and then keeps stable, while that of {\em SElf}$^*$ first decreases, then increases slightly, and finally stays stable. For {\em Elf}$^*$, more data provides more accurate distributions of leading zeros and trailing zeros. The distributions of 800 consecutive samples are close to the ground truths, so the compression ratio of {\em Elf}$^*$ keeps stable when the window size is larger than 800. For {\em SElf}$^*$, when the window size is too small (e.g., smaller than 600), it cannot capture the real distributions accurately. When the window size is too large (e.g., larger than 1,000), the distribution difference of two consecutive windows is large too, leading to a poor compression performance if we share the rules.

As shown in Fig.~\ref{fig:exp:windowsize}(b), with a larger window size, the decompression time of both {\em Elf}$^*$ and {\em SElf}$^*$ keeps stable (their decompression time is almost the same). However, their compression time first drops steeply and then keeps stable. When the window size is smaller,  there are more windows for a given dataset, which triggers more times of rules calculation.

In our datasets, the window size of 1,000 can ensure a good compression performance for {\em SElf}$^*$.

\section{Related Works}\label{sec:related}

This paper aims to compress time series data losslessly, so the lossy compression algorithms~\cite{lazaridis2003capturing, liang2022sz3, lindstrom2014fixed, liu2021decomposed, liu2021high, zhao2021optimizing, zhao2022mdz, liu2021exploring, liu2022dynamic, barbarioli2023hierarchical, li2023lossy, gong2022region, jiao2022toward, eichinger2015time} are outside the discussion scope of this paper. Here, we review the related works from three aspects.

\subsection{Algorithms XORing with Previous Values}
Considering that there is little variation between two consecutive values in a time series, this type of algorithms first performs an XORing operation on them, and then encodes the XORed result to achieve floating-point compression. Based on the assumption that there are both many leading zeros and trailing zeros in the XORed value, Gorilla~\cite{pelkonen2015gorilla} uses 5 bits and 6 bits to encode the numbers of leading zeros and center bits, respectively. Observing the severely unbalanced distribution of leading zeros, Chimp~\cite{liakos2022chimp} encodes the leading zeros with only 3 bits. Besides, inconsistent with Gorilla's assumption, Chimp finds that there are indeed extremely few trailing zeros in the XORed results, so its upgraded version Chimp$_{128}$~\cite{liakos2022chimp} selects from the previous 128 values the one that produces an XORed result with the most trailing zeros, which improves the compression ratio tremendously. However, to accelerate the computation, Chimp$_{128}$ maintains a hash table with a size of $32$KB, which might be not applicable in edge computing scenarios~\cite{mao2017survey, shi2016edge, paparrizos2021vergedb}.

To increase the number of trailing zeros, {\em Elf}~\cite{li2023elf} erases as many last few bits of the floating-point values as possible in a recoverable manner. Besides, it devises a set of elaborated encoding strategies for the erased XORed values. To this end, {\em Elf} achieves an impressive compression ratio with only $\mathcal{O}(1)$ space complexity. Taking into account that the significand counts of values in a time series are usually similar, {\em Elf+}~\cite{Elf+} optimizes the encoding strategy and computation of significand counts, which further enhances the compression ratio and reduces the compression time. However, like Chimp and Chimp$_{128}$, both {\em Elf} and {\em Elf}+ employ a fixed approximation rule with 3 bits for leading zeros. Moreover, they encode the number of center bits with as many as 5 or 7 bits. As a result, there is sill much room for improvement.

\subsection{Algorithms XORing with Predicted Values}
Instead of XORing with the previous value, this type of algorithms first predicts a value based on trained models~\cite{blalock2018sprintz, burtscher2007high, burtscher2008fpc, jensen2018modelardb, jensen2021scalable, ratanaworabhan2006fast, yu2020two}, and then XORs the current value with the predicted one. For example, DFCM~\cite{ratanaworabhan2006fast} and ZFP~\cite{ZFP} predict the value by one predictor, while FPC~\cite{burtscher2008fpc, burtscher2007high} uses two predictors to predict a value each, and then selects the better one. Different from DFCM~\cite{ratanaworabhan2006fast}, ZFP~\cite{ZFP} and FPC~\cite{burtscher2008fpc, burtscher2007high} using the values in the same time series for prediction, the works~\cite{jensen2018modelardb, jensen2021scalable, yu2020two} capture the characteristics of different time series using machine learning models. The model-based algorithms usually face the problem of low efficiency.

\subsection{General Compression Algorithms}
General compression algorithms~\cite{LZ77, LZMA, Xz, alakuijala2018brotli, collet2016zstd, snappy} can also be used to compress floating-point data. For example, LZ77~\cite{LZ77}, with a sliding window, utilizes a dictionary to handle repeated data, resulting in an effective compression. Building upon LZ77, LZMA~\cite{LZMA} introduces a dynamic dictionary whose size can be adjusted according to the characteristic of the input data. With LZMA as the core, Xz~\cite{Xz} adopts multi-threading to improve parallelism, thereby enhancing both efficiency of compression and decompression. Zstd~\cite{collet2016zstd} and Snappy~\cite{snappy} also leverage a trainable dictionary that is generated from a set of samples.
These general compression algorithms can usually reach high compression ratios, but they are relatively time-consuming and not suitable for streaming scenarios.


\section{Conclusion}\label{sec:conclude}
In this paper, we propose {\em Elf}$^*$, which adopts adaptive encoding strategies on the top of {\em Elf+}. Furthermore, we extend {\em Elf}$^*$ to {\em SElf}$^*$ for streaming scenarios. Comprehensive and extensive experiments on 22 datasets provide compelling evidence for the remarkable performance of {\em Elf}$^*$ and {\em SElf}$^*$. Among the streaming algorithms, {\em SElf}$^*$ outperforms {\em Elf+} with a compression ratio improvement of 9.2\% while keeping a similar efficiency, and enjoys improvement of 55.0\% and 26.4\% in compression ratio over Gorilla and Chimp$_{128}$, respectively. In batch processing scenarios, when the block is small, {\em Elf}$^*$ outperforms Xz by 10.1\% in terms of compression ratio but with only 4.8\% compression time. {\em Elf}$^*$ still ranks among the most competitive batch compressors when the block is big. As for future work, we plan to integrate {\em Elf}$^*$ and {\em SElf}$^*$ into hardware to further enhance its transmission capability. 

\section*{Acknowledgment}

This paper is supported by the National Natural Science Foundation of China (61976168, 61872050, 62172066) and China Postdoctoral Science Foundation (2022M720567). We would like to thank Lei Liu and Yuhang Feng for discussing the solutions with us.

\bibliographystyle{IEEEtran}

\footnotesize
\bibliography{IEEEabrv,reference}
\end{document}